%% file: main.tex
\documentclass[10pt, 
aps,
prx,
twocolumn,
noeprint,
longbibliography,
nofootinbib
]{revtex4-2}

\input{preamble.tex}
\graphicspath{{figures/}}
\usepackage{color}

\begin{document}

\title{Rigorous error bounds for dissipative thermal state preparation from weak system-bath coupling}

\author{Christopher Ong}
\author{S. A. Parameswaran}
\author{Benedikt Placke}
\author{Dominik Hahn}
\affiliation{Rudolf Peierls Centre for Theoretical Physics, University of Oxford, Oxford OX1 3PU, United Kingdom}

\begin{abstract}
Thermal state preparation is a central challenge in the simulation of quantum many-body systems. Yet, provably efficient
algorithms for this task were only introduced recently [Chen et al. {\it Nature} {\bf 646}, 561 (2025)]. These algorithms are based on dissipative Lindbladian evolution which exactly fixes the thermal state.  
Controlled and efficient digital simulation of this evolution, although possible in principle, remains out of reach for present-day quantum hardware.
Subsequent work has therefore focused on analog approximations of the proposed Lindbladians via `collision models' with  relatively modest requirements --- a resettable bath of  ancilla qubits whose couplings to the system can be tuned in time-dependent fashion --- while still admitting rigorous fixed-point error bounds. Existing rigorous approaches, however, do not exploit the fact that these constructions generically
implement not only the desired Lindblad dynamics, but also an additional unitary evolution generated by the system Hamiltonian which may aid convergence to the thermal state [Lloyd and Abanin arXiv:2506.21318 (2025)]. Here, we show that this unitary contribution does indeed tighten the fixed-point error bound 
and demonstrate that it is rigorously controlled by the system-bath coupling strength $J$, scaling as $J^2$. 
This demonstrates that the effect of the spurious `Lamb shift' term generated by the system-bath interaction can be controlled by tuning $J$. 
We clarify the role, previously observed, of a randomized implementation in suppressing possible resonances of the drive with the many-body spectrum, and bound the additional variance that this randomization imposes on observables. 
Finally, we numerically study aspects of the protocol which are relevant for its practical realization, such as the mixing time.
\end{abstract}

\maketitle

\section{Introduction}

Thermal state preparation is a central task in the simulation of quantum many-body systems, as it is essential for probing their equilibrium properties. Yet, despite its importance, controlled preparation of the thermal state of a generic local Hamiltonian remains difficult. In the classical setting, Markov chain Monte Carlo methods such as the Metropolis algorithm have been highly successful \cite{metropolis1953,levin2017markov}. In the quantum setting, however, many existing approaches---such as imaginary-time evolution \cite{motta2020}, variational methods \cite{consiglio2023,deshpande2025, ilin2025}, and phase-estimation-based algorithms \cite{poulin2009}---continue to face significant challenges involving some subset of their lack of scalability, absence of rigorous convergence guarantees, and inconvenience of implementation on current quantum devices. 

Another promising approach for thermal state preparation is provided by dissipative algorithms~\cite{PhysRevA.61.022301,Measurmenetsteering2020,zhang2023dissipativequantumgibbssampling,Lin2025Dissipative}. The central idea is to engineer an open-system dynamics that drives the system toward the desired state as the fixed point of the evolution.
Recently, a family of dissipative algorithms has been discovered~\cite{chen2025efficient,CKG23,ding2025efficient}, based on the efficient simulation of a family of Lindbladians that satisfy the Kubo-Martin-Schwinger~(KMS) detailed balance conditions~\cite{agarwal1973,alicki1976detailed,frigerio1977quantum,fagnola2007generators}. These conditions are  a quantum generalization of classical detailed balance, and they imply that the exact fixed point of the Lindbladian is given by the Gibbs state $\rho_\beta = Z_\beta^{-1} e^{-\beta H}$, where $Z_\beta =  \text{Tr}\, e^{-\beta H}$ is the partition function.
These algorithms can therefore be interpreted as a quantum generalization of Markov chain Monte Carlo algorithms~\cite{gilyen2026quantumgeneralizationsglaubermetropolis}.
Although Lindbladians satisfying KMS detailed balance have long been known in the form of `Davies generators'~\cite{davies1974,davies1976}, a key new feature of this recently discovered family is that it is generated by quasi-local jump operators. This not only implies their resource-efficient simulability at least in principle, but also makes some aspects of their theoretical analysis much more tractable. For example, it allows one to derive estimates on the mixing time, i.e. the convergence time to equilibrium, in certain regimes, such as at high temperatures~\cite{rouze2024efficient,rouze2024optimal,bakshi2025dobrushinconditionquantummarkov}, for weak interactions~\cite{tong2024poly_mixing,smid2025poly_mixing}, in one-dimensional systems~\cite{bergamaschi2026fastmixingquantumspin}, and on expander graphs~\cite{placke2024slow_mixing,rakovszky2024bottleneck,garmarnik2024slow_mixing}.

Despite these major advances, implementing these exact Gibbs samplers remains challenging in practice. Their jump operators are defined in terms of weighted Heisenberg evolutions of local operators and are therefore not straightforward to realize directly. The first theoretical proposals for their implementation~\cite{CKG23,ding2025efficient} were based on a block-encoding of the Lindbladian, which is beyond the reach of present-day quantum devices and analog platforms. Other implementation strategies have been proposed, including randomized compilation~\cite{Chen_2025}, variational compilation of the jump operators~\cite{hahn2025efficientquantumgibbssampling}, and approaches that leverage additional properties of the system, such as assumptions motivated by the eigenstate thermalization hypothesis~\cite{brunner2024lindblad_engineering,shtanko2023preparingthermalstatesnoiseless}.

A more physically motivated approach is to implement these Gibbs samplers \emph{approximately} using `collision models' of a type familiar from the open quantum systems literature.  In this setting, the system is weakly coupled to bath degrees of freedom represented by ancilla qubits, initially in some unentangled (product) state. After each interaction period, the ancillas are `reset', i.e. reinitialized in the product state.
This repeated reset enables the effective removal of entropy from the system. For a specific choice of time-dependent system-bath dynamics and in a suitably defined weak-coupling limit~[cf. Fig.~\ref{fig:protocol}], such models can be designed to approximately realize a Gibbs sampler~\cite{hahn2026efficientquantumthermalstate,ding2025endtoendefficientquantumthermal,lloyd2506quantum,ScandiThermalization}. For specific choices of parameters, this approximation can even extend beyond the weak-coupling regime~\cite{wang2026lindbladdynamicsrigorousguarantees}.

Although this class of protocols is appealing as it can be implemented using only local interactions, its rigorous analysis is challenging. At leading order in the system-bath coupling strength, the effective dynamics consists of a Lindbladian evolution, followed by a unitary evolution with the system Hamiltonian. 
While the dissipative part of the Lindbladian evolution can be chosen to match that of an exact Gibbs sampler, their coherent parts differ due to an effective back-action of the system-bath dynamics, often referred to as the `Lamb shift'. 
As a consequence, the thermal state is no longer an exact fixed point of the Lindbladian evolution. 
Previous analyses~\cite{hahn2025efficientquantumgibbssampling,ding2025endtoendefficientquantumthermal} studied the fixed-point error resulting from this Lamb shift, and derived bounds in terms of the mixing time and structural properties of the Lindbladian. However, by neglecting the effect of the unitary system evolution (in effect, by `rewinding' this via backward time-evolution \cite{hahn2026efficientquantumthermalstate}), parts of their bounds were insensitive to the system-bath coupling strength.

Ref.~\onlinecite{lloyd2506quantum} identified another key mechanism which may control the fixed-point error in the weak-coupling limit.
In particular, their analysis suggests that including the unitary evolution under the system Hamiltonian in the error analysis can significantly reduce the fixed-point error, which can be controlled and made arbitrarily small by tuning the system-bath coupling strength. They further introduce a randomization of the evolution time to suppress many-body resonances excited by the periodic system-bath evolution and thus remove possible divergences in the error bound.
While this picture is supported by perturbative arguments, our main contribution is to establish it rigorously and sharpen the resulting analysis and error bounds.
Concretely, we prove that the fixed-point error of the randomized protocol scales quadratically in the system-bath coupling strength, and hence can be made arbitrarily small by tuning this coupling strength. This implies an overall resource scaling of order $\mathcal{O}(\epsilon^{-1})$ for preparing the fixed point to accuracy $\epsilon$, which is better than existing protocols that control the fixed-point error by other means.

\begin{figure}
    \centering
    \includegraphics[width=\linewidth]{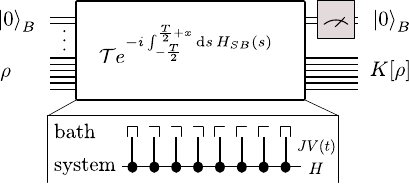}
    \caption{Setup of the channel $\mathcal{K}[\rho]$ to prepare a thermal state. At each step, a set of ancillas is initialized in the product state $\ket{0}_B$. The system and ancilla qubits are evolved in an interval $[-\frac{T}{2},\frac{T}{2}+x]$ with the Hamiltonian $H_{SB}=H\otimes \mathbb{I}_B+ J V(t)$, with the interaction $V(t)$ defined in Eq.~\eqref{eq:interaction}. The final time $\frac{T}{2}+x$ is drawn from a distribution $p(x)$~[cf. Eq.~\eqref{eq:px}] to suppress resonances of the system-bath interaction. The ancilla qubits are reset to the initial state $\ket{0}_B$ at the end of the cycle.
    The fixed point $\rho_\mathrm{fix}$ of the channel $\mathcal{K}[\rho]$ is approximately the thermal state $\rho_\beta=Z_\beta^{-1} e^{-\beta H}$, with a rigorous error bound scaling as $J^2$.}
    \label{fig:protocol}
\end{figure}

Our rigorous bounds are obtained for the fictitious channel that is the average over various realizations. 
Practical realizations of the protocol would however be stochastic: in each time step, one evolves for an evolution time sampled at random from a specified distribution. While this procedure is `self-averaging', in that the expectation value over such sequences reproduces the desired averaged channel, the randomization introduces additional variance into measurements of observables, thereby increasing the number of samples required to estimate expectation values. To our knowledge, this additional sampling variance has not been previously noted or studied for quantum Gibbs samplers. Our second main result is to quantify the additional overhead incurred by the randomization step and demonstrate that it does not severely impact the efficiency of the protocol.  Specifically, we derive bounds on the extra variance generated by repeated applications of the randomized protocol.
We further support this analysis with perturbative arguments showing that omitting the randomization step worsens the fixed-point error, even in the presence of only sparse resonances.

Our results provide a comprehensive and rigorous analysis of several practically relevant properties of the thermal state preparation protocol, and in particular show how the coupling strength and other parameters can be tuned to achieve a prescribed accuracy. Taken together, these results establish the algorithm firmly as a promising candidate for implementation on near-term digital and analog quantum platforms.

The remainder of this work is organized as follows. In Section~\ref{sec:Gibbs samplers}, we review the aspects of quantum Gibbs samplers that are most relevant for the present work, before introducing the protocol in Section~\ref{sec:protocol}. In Section~\ref{sec:methods}, we highlight several properties of the channel that are used in subsequent derivations of error bounds. In Section~\ref{sec:Results}, we state our main analytical results concerning the fixed-point error and the variance of observable measurements. In Section~\ref{sec:Numerical results}, we present a numerical analysis of the scaling of the spectral gap of the channel with the coupling strength, the role of the randomization step, and the variance of observables. Finally, Section~\ref{ref:Discussion} contains a discussion of the implications of our results, as well as several open questions.

\section{Exact Gibbs samplers}\label{sec:Gibbs samplers}

We wish to prepare the state $\rho_\beta = Z_\beta^{-1} e^{-\beta H}$, where $H$ is a given local Hamiltonian and $Z_\beta = \Tr e^{-\beta H}$ is the partition function. This can be done by means of constructing a Lindbladian $\mathcal L$ which fixes this state as its unique steady state $\mathcal L[\rho_\beta] = 0$.

While the general construction of such a `quantum Gibbs sampler' for a long time remained an open problem, following the breakthrough of Ref.~\onlinecite{chen2025efficient} there now exist a number of different constructions in the literature~\cite{ding2025efficient,bakshi2025dobrushinconditionquantummarkov,chen2025efficient,CKG23,gilyen2026quantumgeneralizationsglaubermetropolis}. 
Here, we use the construction of Ref.~\onlinecite{ding2025efficient}. In the following, we focus on aspects of the construction which are important for our purposes. For a full exposition of KMS detailed balance and quantum Gibbs samplers, we refer the reader for example to Refs.~\onlinecite{hahn2026efficientquantumthermalstate,ding2025efficient,Guo2025designingopen}.

Consider the following family of Lindbladians:
\begin{align}\label{eq:Gibbs}
    \mathcal{L}_\mathrm{DB}[\rho]=-i[G_{\rm{DB}},\rho]+\sum_{a\in\mathcal{A}} L_a \rho L^\dagger_a-\frac{1}{2}\{L_a^\dagger L_a,\rho\},
\end{align}
with 
\begin{align}\label{eq:jump}
L_a=\int_{-\infty}^\infty \rd t f(t)e^{iHt}A_a e^{-iHt}=\int_{-\infty}^\infty \rd t f(t)A_a(t).
\end{align}
Here, $A_a(t)$ denotes the Heisenberg evolution $A_a(t)=e^{iHt}A_a e^{-iHt}$ of an operator $A_a$, 
and the filter function $f(t)$ is chosen such that it satisfies a particular symmetry 
relation 
\begin{align}\label{eq:ftdefinition}
f(t)=\frac{1}{2\pi}\int \rd \nu \,q(\nu) e^{\frac{\beta \nu}{4}}\e^{i\nu t}, \text{ with } q(-\nu)=q^*(\nu).
\end{align}
Using the Bohr frequency representation,
\begin{align}
    A_\nu=\sum_{\substack{\omega_1, \omega_2 \in {\rm Spec}(H) \\\omega_1-\omega_2=\nu}}\ket{\omega_1}\braket{\omega_1|A|\omega_2}\bra{\omega_2},
\end{align}
where $\mathrm{Spec}(H)$ denotes the set of eigenvalues of $H$, the coherent part can then be expressed as
\begin{align}\label{eq:coherent}
    G_{\rm{DB}}=\frac{i}{2}\sum_{a\in\mathcal A} \sum_{\nu\in\mathfrak{B}} \tanh(\frac{\beta \nu}{4}) (L^\dagger_a L_a)_\nu,
\end{align}
with $\mathfrak{B}$ denoting the set of Bohr frequencies of the Hamiltonian $\mathfrak{B} = \{ \omega_1 - \omega_2: \omega_1, \omega_2 \in \mathrm{Spec}(H)\}$.

It can be shown~\cite{ding2025efficient} that the jump operators in Eq.~\eqref{eq:jump} as well as the coherent part in Eq.~\eqref{eq:coherent} are quasi-local if the filter function $f(t)$ decays faster than a power law. 

The coherent part and jump operators are carefully constructed such that the Lindbladian in Eq.~\eqref{eq:Gibbs} fulfills KMS detailed balance~\cite{CKG23,agarwal1973,alicki1976detailed,frigerio1977quantum,fagnola2007generators}.
Given the superoperator 
\begin{align}
    \Gamma[\cdot]=\rho_\beta^{1/2}(\cdot)\rho_\beta^{1/2},
\end{align}
this relation can be expressed as
\begin{align}\label{eq:KMS detailed balance}
    \mathcal{L}^\dagger_\mathrm{DB}=\Gamma^{-1}\circ \mathcal{L}_\mathrm{DB} \circ \Gamma.
\end{align}
Here, the adjoint operator is defined with respect to the Frobenius inner product
$\Tr[\mathcal{L}[\rho]\mathcal{O}]=\Tr[\rho\mathcal{L}^\dagger[\mathcal{O}]]$. Trace preservation under Lindbladian evolution implies $e^{t\mathcal{L}^\dagger}[\mathbb{I}]=\mathbb{I}$ and therefore $\mathcal{L}^\dagger[\mathbb{I}]=0$.

Importantly, KMS detailed balance [Eq.~\eqref{eq:KMS detailed balance}], directly implies that $\rho_\beta$ is a fixed point of the dynamics under $\mathcal{L}_\mathrm{DB}$:
\begin{align}
    \mathcal{L}_\mathrm{DB}[\rho_\beta]=\rho_\beta^{1/2}\mathcal{L}^\dagger_\mathrm{DB}[\rho_\beta^{-1/2}\circ \rho_\beta\circ \rho_\beta^{-1/2}]\rho_\beta^{1/2}=0.
\end{align}

\section{Setup}\label{sec:protocol}

In this section, we introduce the setup of the thermal state preparation algorithm studied in this work. As shown in Fig.~\ref{fig:protocol}, the protocol is defined on a system of \(n_S\) qubits coupled to a bath of \(n_B\) ancilla qubits. One step of the protocol consists of initialization of the bath qubits to a product state $\rho_B^0 = \ketbra{0}{0}_B$, followed by evolution with a time-dependent, joint system-bath Hamiltonian $H_{SB}(t)$, and then reset of the bath qubits to $\rho_B^0$.
Note that in principle, $n_B=1$ is expected to be sufficient for generic Hamiltonians~\cite{lloyd2024quasiparticle}; however, the convergence time of the algorithm is expected to improve using $n_B=\mathcal{O}(n_S)$ ancilla qubits [as discussed below, we will aim to approximate a given exact Gibbs sampler, in which case we need one ancilla per jump operator $L_a$ in Eq.~\eqref{eq:Gibbs}].

The system-bath Hamiltonian used during the protocol is given by
\begin{align}\label{eq:HSB}
    H_{SB}(t) = H \otimes {\mathbb I}_B + {\mathbb I}_S \otimes H_B + J V(t).
\end{align}
Here, $H$ is the system Hamiltonian, $H_B$ is the bath Hamiltonian, and the term $J V(t)$ is the interaction between system and bath, where $J$ denotes the strength of this coupling. 

The interaction term $V(t)$ takes the form
\begin{align}\label{eq:interaction}
    V(t) = \sum_{a\in\mathcal A} f(t) A_a \otimes B_a^\dagger + f^*(t) A_a^\dagger \otimes B_a.
\end{align}
The set of jump operators $\mathcal{A}$ can be chosen arbitrarily, as long as they are closed under Hermitian conjugation ($A_a\in \mathcal{A}\Rightarrow A^\dagger_a\in \mathcal{A}$) and guarantee irreducibility of the Lindbladian evolution. (A convenient sufficient choice independent of the Hamiltonian $H$ would be, for instance, the single-qubit Pauli operators $X$ and $Y$ on each site.)

Each jump operator is coupled to a distinct bath qubit, where
$B=\frac{1}{2}(X_B-iY_B)$ and $B^\dagger=\frac{1}{2}(X_B+iY_B)$ are the lowering and raising operators of the corresponding bath qubit, respectively.
The filter function $f(t)$ is chosen such that it fulfills the relation defined in Eq.~\eqref{eq:ftdefinition} for exact Gibbs samplers.
For the rest of this work, for concreteness, we use 
\begin{align}\label{eq:f}
    f(t)=\sqrt{\frac{2}{\pi \sigma^2}} e^{-\frac{2}{\sigma^2}(t-\frac{i \beta}{4})^2}.
\end{align}
However, we expect our results to hold more generically, as long as $f(t)$ fulfills Eq.~\eqref{eq:ftdefinition} and decays sufficiently fast. 

The dynamics of the bath is chosen to be trivial $H_B={\mathbb I}_B$. Although we do not use it here, an equivalent construction can instead introduce a bath Hamiltonian and choose the interaction amplitude $f(t)$ to be real, which may be easier to implement in practice~\cite{hahn2026efficientquantumthermalstate,lloyd2506quantum}.

We define the following propagator $U(t)$
\begin{align}\label{eq:propagatorcomp}
    U\left(t+\frac{T}{2}\right)=\mathcal{T}e^{-i \int_{-\frac{T}{2}}^t \rd s \, H_{SB}(s)},
\end{align}
where $\mathcal T$ denotes time ordering, and $T$ is an interaction timescale fixed as a parameter of the algorithm.
At each step of the algorithm, the ancilla qubits are initialized in the product state $\rho_B^0=\ket{0}\bra{0}_B$ at time ${t=-\frac{T}{2}}$.  The combined system is then evolved under Eq.~\eqref{eq:HSB} until time $t = \frac{T}{2}+x$, where the shift $x$ is drawn from a Gaussian distribution with randomization timescale $T_0$
\begin{align}\label{eq:px}
 p(x)=\frac{1}{\sqrt{\pi T_0^2}} e^{-\left(\frac{x}{T_0}\right)^2}.
\end{align}
At the end of the cycle, the ancilla qubits are reset back to $\rho_B^0=\ket{0}\bra{0}_B$. 
For a given $x$, one step of the protocol thus applies a channel $\mathcal{K}_x[\rho]$ on the system, which takes the form
\begin{align}\label{eq:single instance}
     \mathcal{K}_x[\rho]=\Tr_B \left[U\left(T+x \right)(\rho \otimes \rho_B^0) U^\dagger\left(T+x \right)\right].
\end{align}
 
We define $\mathcal{K}[\rho]$ as the average of $\mathcal{K}_x[\rho]$ over the distribution of evolution time shifts $x$,
\begin{align}\label{eq:channelrandomized}
    \mathcal{K}[\rho]=\int_{-\infty}^{\infty} \rd x\, p(x) \mathcal{K}_x[\rho].
\end{align}

As was suggested in Ref.~\onlinecite{lloyd2506quantum}, the averaging over different evolution times acts effectively as a `dephaser' that removes possible resonances of the drive with possible transitions in the system's many-body spectrum and eliminates potential divergences in the fixed-point error expansion that plague the deterministic limit of the protocol. This point will be discussed more carefully in Sec.~\ref {subsec:Effectresonances}.

As mentioned in the introduction, although in practice one implements the randomized protocol, to obtain rigorous bounds it is more convenient to  take a two-step approach.  First, we invoke the self-averaging property of the randomized evolution to focus on the fixed-point error of the average channel $\mathcal{K}$. However, the real randomized protocol incurs errors in observables relative to $\mathcal{K}$, controlled by the variance of the time-step distribution; our second step is to bound this error. Taken together, these serve to rigorously bound the error of the true randomized protocol.

The convergence of the channel toward the fixed point is characterized by its mixing time $\tau_{\mathcal{K},\mathrm{mix}}(\epsilon)$. It is defined as the minimum number of iterations of the channel $\mathcal{K}[\rho]$ required so that any initial state is guaranteed to be within $\epsilon$ of its fixed point $\rho_\mathrm{fix}$.
\begin{equation}\label{eq:def_tau_mix}
    \tau_{\mathcal{K}, \mathrm{mix}}(\epsilon) = \min \left\{ \, M \in \mathbb{N} \,\middle|\, \sup_{\rho} \left\| \mathcal{K}^{M}[\rho] - \rho_\mathrm{fix} \right\|_{1} \le \epsilon \right\}. 
\end{equation}
Here, the distance to the fixed point is quantified using the trace norm,
\begin{align}
    \norm{\rho}_1=\Tr[\sqrt{\rho\rho^\dagger}].
\end{align}

\section{Methods and Heuristic Overview}\label{sec:methods}

In this section, we provide some results on the properties of the channel $\mathcal{K}[\rho]$ that are necessary for bounding the fixed-point error. We first show, in Sec.~\ref{subsec:ClosenessLindblad}, that the channel is well-approximated by a Lindbladian evolution followed by a unitary evolution with the system Hamiltonian. Next, we construct an approximate fixed point of this combined evolution in Sec.~\ref{subsec:Approx} and show that its distance to $\rho_\beta$ is controlled by the coupling strength $J$. 
In Sec.~\ref{subsec:Heuristicpicture}, we provide a heuristic explanation for the role of the unitary evolution in controlling the distance of the fixed point to the thermal state.
Finally, in Sec.~\ref{subsec:Effectresonances}, we heuristically study the effect of resonances between the driving period of the system-bath coupling and transitions in the many-body spectrum of $H$, which demonstrates the necessity of the randomized evolution time. 

\subsection{Approximation of the channel}\label{subsec:ClosenessLindblad}

The average channel $\mathcal{K}[\rho]$ almost implements a Lindblad evolution under a Gibbs sampling algorithm, followed by unitary dynamics generated by $H$. Concretely, consider the Lindbladian
\begin{align}\label{eq:LLSdefinition}
    \mathcal L_\mathrm{LS}[\rho]=-i[\Delta G,\rho]+\mathcal L_\mathrm{DB}[\rho],
\end{align}
with
\begin{align}
    \Delta G=G_{\mathrm{LS}}-G_{\mathrm{DB}}.
\end{align}
The Lamb shift term $G_{\mathrm{LS}}$ is given by
\begin{multline} \label{eq:GLS}
    G_{\mathrm{LS}} = -\frac{1}{2i} \sum_{a\in\mathcal A} \int_{-\infty}^{\infty} \rd t_1 \int_{-\infty}^{\infty} \rd t_2 \, \text{sign}(t_1 - t_2)\\ f^*(t_2) f(t_1) A^\dagger_{a}(t_2) A_{a}(t_1).
\end{multline}
This Lindbladian agrees with the exact Gibbs sampler in Eq.~\eqref{eq:Gibbs}, up to a different coherent term given by $G_{\mathrm{LS}}$.
Furthermore, denote unitary evolution by the system alone by $U_0(t) = e^{-iHt}$ and consider the channel
\begin{align}\label{eq:channelinfty}
    \mathcal{K}_{\mathrm{LS}}[\rho]= \int_{-\infty}^{\infty} \rd x\, p(x) U_0(T+x)\,\e^{J^2 \mathcal L_\mathrm{LS}}\rho \, U_0^\dagger(T+x).
\end{align}
The difference between $\mathcal{K}_{\mathrm{LS}}[\rho]$ and the average channel $\mathcal{K}[\rho]$ is then bounded by 
\begin{multline}\label{eq:channel}
    \norm{\mathcal{K}-\mathcal{K}_{\mathrm{LS}}}_{1\rightarrow 1} \\= \mathcal{O} \left( J^2 n_B e^{\frac{\beta^2}{4\sigma^2}} e^{-\frac{T^2}{2\sigma^2 + 4 T_0^2}} + J^4 n_B^2 e^{\frac{\beta^2}{2\sigma^2}} \right),
\end{multline}
where the difference between channels is characterized by the induced norm
\begin{align}
    \norm{\mathcal{M}}_{1\rightarrow 1}=\sup_{\rho}\frac{\norm{\mathcal{M}[\rho]}_1}{\norm{\rho}_1}.
\end{align}

The result follows directly from a Dyson series expansion and is derived in App.~\ref{app:boundchanneldist}. 
The first error term originates from truncating Gaussian tails of the filter function due to the finite-time system-bath evolution in each iteration step~[cf. Fig.~\ref{fig:protocol}]. This term is exponentially suppressed in the length of the time evolution $T$.
The second term bounds contributions from higher-order terms of the Dyson series expansion. 

From the definition of Eq.~\eqref{eq:channelinfty}, we see that \(J^2\) plays the role of a time step of the evolution with an effective Lindbladian $\mathcal L_{\rm LS}$. Since the effective Lindbladian is the only dissipative part of the protocol, this suggests (although does not show) that the mixing time of the channel $\mathcal K_{\rm LS}$, $\tau_{\mathrm{mix}}(\epsilon) \equiv \tau_{\mathcal{K}_{\mathrm{LS}}, \mathrm{mix}}(\epsilon)$, scales as the inverse of this step size $J^{-2}$.
It is then natural to introduce a rescaled mixing time ${t_{\mathrm{mix},J}(\epsilon)}$ of the channel
\begin{equation}\label{eq:scsaling}
    t_{\mathrm{mix},J}(\epsilon) = J^2 \tau_{\mathrm{mix}}(\epsilon).
\end{equation}

Although we expect that $t_{\mathrm{mix},J}$ approaches a constant for $J\rightarrow 0$, and we provide numerical evidence for this in Sec.~\ref {sec:Numerical results}, we stress that we do not have a general proof that $\lim_{J\rightarrow 0 }t_{\mathrm{mix},J}(\epsilon)$ exists.

There exist some cases where the existence of this limit can be rigorously established, when $2 \norm{\rho_\beta^{1/4}\Delta G \rho_\beta^{-1/4}}_\infty$ is smaller than the spectral gap of the Gibbs sampler $\mathcal{L}_\mathrm{DB}$~\cite{ding2025endtoendefficientquantumthermal}, but establishing this limit more generally remains an interesting direction for future work.  
In Sec.~\ref {sec:Numerical results}, we provide numerical evidence suggesting that this scaling persists even when these existing rigorous results do not apply.

Finally, we note that if $T$ is chosen sufficiently large and $J$ sufficiently small so that
\begin{equation}
    \tau_{\mathrm{mix}}\!\left(\frac{\epsilon}{2}\right)\,
    \norm{\mathcal{K}-\mathcal{K}_{\mathrm{LS}}}_{1\to1}
    \leq \frac{\epsilon}{2},
\end{equation}
then the mixing times $\tau_{\mathcal{K}, \mathrm{mix}}$ of $\mathcal{K}[\rho]$ and $\tau_{\mathrm{mix}}$ of $\mathcal{K}_{\mathrm{LS}}[\rho]$ are related by~[App. D in Ref.~\onlinecite{ding2025endtoendefficientquantumthermal} and Lemma~\ref{lem:closemixingtimes}, App.~\ref{app:fixedpointsclose}]
\begin{align}
    \tau_{\mathcal{K}, \mathrm{mix}}(2\epsilon)
    \leq
    \tau_{\mathrm{mix}}\!\left(\frac{\epsilon}{2}\right).
\end{align}

\subsection{Approximate fixed point of $\mathcal{K}_{\mathrm{LS}}[\rho]$}\label{subsec:Approx}

In this subsection, we show that the channel $\mathcal{K}_{\mathrm{LS}}[\rho]$ has an approximate fixed point $\tilde{\rho}$ with the following properties:
\begin{enumerate}
    \item $\quad \mathcal{K}_{\mathrm{LS}}[\tilde \rho] = \tilde \rho + \mathcal{O}(J^4)$
    \item $\quad \norm{\tilde{\rho}-\rho_\beta}_1 = \mathcal{O}(J^2)$   
\end{enumerate}

This fact will be crucial to obtain our rigorous fixed-point error bound in Sec.~\ref{subsec:results}. 

An ansatz for the approximate fixed point $\tilde \rho$ can be motivated by degenerate perturbation theory as done in Ref.~\onlinecite{lloyd2506quantum}. While degenerate perturbation theory is not a controlled order-by-order expansion, the first-order expression is sufficient for our purposes.
We outline the main steps here.

The strategy is to match the relation $\mathcal{K}_{\mathrm{LS}}[\tilde{\rho}]=\tilde{\rho}$ [cf. Eq.~\eqref{eq:channelinfty}] with an ansatz $\tilde{\rho} = \rho^{(0)} + J^2\sigma$ up to order 
$J^2$.
This is done in the eigenbasis of the Hamiltonian~($H\ket{\phi_a}=E_a\ket{\phi_a}$, $A_{ab}=\braket{\phi_a|A|\phi_b}$, $\omega_{ab}=E_a-E_b$).

Equating leading-order terms of order $J^0$ yields
\begin{align}\label{eq:perturbativeJ0}
    e^{-\frac{1}{4}\omega_{ab}^2T_0^2-i\omega_{ab}T} \rho^{(0)}_{ab}=\rho^{(0)}_{ab},
\end{align}
which implies that ${\rho}^{(0)}$ is diagonal in the eigenbasis of $H$. 
To determine ${\rho}^{(0)}$ explicitly, it is necessary to consider the $J^2$ contributions for the diagonal elements.
These give
\begin{align}
(\mathcal L_\mathrm{LS}[\rho^{(0)}])_{aa}=0.
\end{align}
It follows that $\rho^{(0)}=\rho_\beta$, since $([\Delta G,\rho^{(0)}])_{aa}=0$. This shows that this ansatz satisfies 
\begin{equation}\norm{\tilde{\rho}-\rho_\beta}_1 = \mathcal{O}(J^2)\label{eq:ansatzerrorboundJ}
\end{equation}
as long as $\norm{\sigma}_1$ is bounded. The relation for the off-diagonal elements $a\neq b$ at order $J^2$ is given by
\begin{align}\label{eq:perturbativeJ2}
   e^{-i\omega_{ab}T}\left(-i[\Delta G,\rho_\beta]_{ab}+e^{-\frac{1}{4}\omega_{ab}^2T_0^2} \sigma_{ab}\right)=\sigma_{ab}.
\end{align}
Solving for $\sigma_{ab}$ yields,
\begin{align}\label{eq:sigmaexpression}
    \sigma_{ab}=\frac{-i e^{-i\omega_{ab}T}[\Delta G,\rho_\beta]_{ab}}{1-e^{-\frac{1}{4}\omega_{ab}^2 T_0^2-i\omega_{ab}T}}.
\end{align}

We can now substitute the ansatz $\tilde{\rho}=\rho_\beta+J^2\sigma$ as defined above back into $\mathcal{K}_{\mathrm{LS}}$, and rigorously show that (1) 
$\mathcal{K}_{\mathrm{LS}}[\tilde{\rho}] =\tilde{\rho} + \mathcal{O}(J^4)$, and (2) its distance from $\rho_\beta$ is bounded by [App. \ref{app:boundG}] 
\begin{align}\label{eq:boundapprox}
    \norm{\tilde{\rho} - \rho_\beta}_1=J^2\norm{\sigma}_1 \leq J^2 \mathrm{poly}\left(n_B,\frac{\beta}{\sigma},\frac{T}{\sigma},\frac{T}{T_0}\right)e^{\frac{9\beta^2}{4\sigma^2}}.
\end{align}

\subsection{Heuristic picture}\label{subsec:Heuristicpicture}

The explicit form of $\sigma$ in Eq.~\eqref{eq:sigmaexpression} as well as the bound in Eq.~\eqref{eq:boundapprox} suggests that the presence of the additional propagator $U_0(T+x)$ in the channel $\mathcal{K}_{\mathrm{LS}}[\rho]$ is crucial to ensure  $\norm{\tilde{\rho}-\rho_{\beta}}_1=\mathcal{O}(J^2)$.
To see this more concretely it is useful to separate the action of the channel into two distinct steps in line with our discussion above, namely `fixing the diagonal algebra' [via Eq.~\eqref{eq:perturbativeJ0}] and `selecting the Gibbs weights' [via Eq.~\eqref{eq:perturbativeJ2}] within it. When the unitary step is eliminated (as achieved, e.g. by the `rewinding' step  in Ref.~\onlinecite{hahn2026efficientquantumthermalstate})  the channel at $\mathcal{O}(J^0)$ is just the identity so {\it every} possible operator is a fixed point. Then, at $\mathcal{O}(J^2)$ the whole Lindbladian correction has to choose the fixed point, and so we have to simply diagonalize the $\mathcal{O}(J^2)$ problem: evidently, there can be no sense in which this is a perturbative procedure in $J^2$. 

Contrast this to the case when the unitary step is implemented: here the zeroth order map is instead coherent evolution with $H$; ignoring resonances for the moment, the fixed operators of this are, as we have argued, diagonal in the eigenbasis of $H$. 
So already at $\mathcal O(J^0)$ the unitary part removes the off-diagonal degeneracy of the fixed point, and then within the set of diagonal states, there is a clear difference between the action of the different terms at $\mathcal O(J^2)$. The detailed-balance piece $\mathcal{L}_{\rm DB}$ fixes the weights within the diagonal sector to agree with $\rho_\beta$, and this happens for \emph{any} finite $J$, even as $J \to 0^+$. The coherent perturbation (the difference between the Lamb shift and the `correct' coherent part) is off-diagonal, and, since the off-diagonal elements of $\tilde \rho$ are fixed to zero at $\order{J^0}$, it can only induce coherences perturbatively in $J^2$.

Finally, Eq.~\eqref{eq:sigmaexpression} highlights the role of the probabilistic choice of the final time $t = \frac{T}{2}+x$. In the deterministic case $T_0=0$, the denominator can vanish whenever $\omega_{ab}T=2 \pi k$ with $k$ an integer. This characterizes many-body resonances induced by the time-dependent system-bath interaction. The randomized 
evolution time eliminates them by dephasing.

We see, therefore, that the additional unitary evolution plays a crucial role in privileging the commutant of the algebra generated by $H$ from within the vastly bigger full operator Hilbert space already at $\mathcal O(J^0)$, enabling a controlled treatment in $J^2$, and in turn necessitates randomized evolution to eliminate resonances in this procedure.

\subsection{Effect of many-body resonances}~\label{subsec:Effectresonances}

As discussed in the previous subsection, resonances cause divergences in Eq.~\eqref{eq:sigmaexpression} if $T_0 = 0$ and $\omega_{ab}T=2 \pi k$ with $k$ an integer. It is natural to ask whether these resonances remain isolated, affecting only the (few) matrix elements directly corresponding to the resonant transitions. We show here that they do not: resonances can also modify populations in non-resonant subspaces. 

To illustrate this, we construct an approximate fixed point $\rho_0$ of the channel $\mathcal{K}_0[\rho]$ corresponding to the deterministic choice $T_0=0$. To this end, we follow the same arguments that led to our ansatz for the approximate fixed point of $\mathcal{K}_{\mathrm{LS}}[\rho]$ in Sec.~\ref{subsec:Approx}.

As before, the channel $\mathcal{K}_0[\rho]$ can be approximated by a Lindblad evolution and a unitary evolution,
\begin{equation}
    \mathcal{K}_0[\rho]= e^{-iHT} \rho e^{iHT} + J^2 e^{-iHT} \mathcal{L}_{\mathrm{LS}}[\rho] e^{iHT} + \mathcal{O}(J^4).
\end{equation}
We again match the relation 
$\mathcal{K}_0[\rho_0]=\rho_0$ with the ansatz 
$\rho_0 = \rho^{(0)} + J^2 \rho^{(1)}$. At order $J^0$, this gives
\begin{equation}\label{eq:perturbativeJ0degen}
    e^{-iHT} \rho^{(0)} e^{iHT} = \rho^{(0)}.
\end{equation}
Note that already in this order there is a difference between the resonant and non-resonant case. At a resonance where $\omega_{ab}T=2 \pi k$, Eq.~\eqref{eq:perturbativeJ0degen} is satisfied even for non-diagonal $\rho^{(0)}$.

In order to analyze the matrix elements in the eigenbasis of $H$, we split it into a resonant subspace $\{a,b\}$ (where $e^{i \omega_{ab} T} = 1$) and a non-resonant subspace $\{i,j\}$ (where $e^{i \omega_{ij} T}\neq1$). Eq.~\eqref{eq:perturbativeJ0degen} then implies that $\rho^{(0)}$ has zero elements,
\begin{equation}
    (\rho^{(0)})_{ai} = 0,\ 
    (\rho^{(0)})_{ij} = 0,
\end{equation}
and nonzero elements,
\begin{equation}
    (\rho^{(0)})_{aa} \neq 0, \ 
    (\rho^{(0)})_{ii} \neq 0, \ 
    (\rho^{(0)})_{ab} \neq 0. 
\end{equation}
Notably, off-diagonal elements $(\rho^{(0)})_{ab}$ in the resonant subspace are now in general allowed to be nonzero. To fully determine $\rho^{(0)}$ we have to consider the order $J^2$ contributions in this subspace.
At second order, it is given by
\begin{equation} \label{eq:perturbativeJ2degen}
    e^{-iHT}\big(\mathcal{L}_{\mathrm{LS}}[\rho^{(0)}]+\rho^{(1)}\big)e^{iHT} = \rho^{(1)}
\end{equation}
which implies that $\mathcal{L}_{\mathrm{LS}}[\rho^{(0)}]$ has zero elements,
\begin{align}
\begin{split}
(\mathcal{L}_{\mathrm{LS}}[\rho^{(0)}])_{aa} &= 0,\\ 
(\mathcal{L}_{\mathrm{LS}}[\rho^{(0)}])_{ii} &= 0,\\
(\mathcal{L}_{\mathrm{LS}}[\rho^{(0)}])_{ab} &= 0,
\end{split}
\end{align}
and nonzero elements,
\begin{align}
(\mathcal{L}_{\mathrm{LS}}[\rho^{(0)}])_{ai} \neq 0,\ 
(\mathcal{L}_{\mathrm{LS}}[\rho^{(0)}])_{ij} \neq 0.
\end{align}

The structures of $\rho^{(0)}$ and $\mathcal{L}_{\mathrm{LS}}[\rho^{(0)}]$ are summarized below for the case of a 4-level system, where the box indicates the resonant space $\{a,b\}$, and matrix elements are represented by `0' if zero or `$\times$' if generally nonzero.
\begin{equation} \label{eq:degeneratePTresult}
\rho^{(0)} \sim \begin{pmatrix} \boxed{\begin{matrix} \times & \times\\ \times & \times \end{matrix}} & \begin{matrix} 0 & 0\\ 0 & 0 \end{matrix} \\[1em] \begin{matrix} 0 & 0\\ 0 & 0 \end{matrix} & \begin{matrix} \times & 0\\ 0 & \times \end{matrix} \end{pmatrix}, \ 
\mathcal{L}_{\mathrm{LS}}[\rho^{(0)}] \sim \begin{pmatrix} \boxed{\begin{matrix} 0 & 0\\ 0 & 0 \end{matrix}} & \begin{matrix} \times & \times\\ \times & \times \end{matrix} \\[1em] \begin{matrix} \times & \times\\ \times & \times \end{matrix} & \begin{matrix} 0 & \times\\ \times & 0 \end{matrix} \end{pmatrix}
\end{equation}

The state $\rho^{(0)}$ is then determined implicitly as the solution of Eq.~\eqref{eq:degeneratePTresult}.
Already its structure is different from the non-resonant case, and in particular $\rho^{0}$ is not fixed to be diagonal. Furthermore, 
in Sec.~\ref{subsec:Dephasing}, we obtain the solution to Eq.~\eqref{eq:degeneratePTresult} numerically for a small system. The results show that, at resonance, the matrix elements of the approximate fixed point $\rho_0$ differ from those of $\rho_\beta$ even in the limit $J\to 0^+$, and that this discrepancy persists even within non-resonant blocks.
Thus, the effects of resonances are not even confined to resonant blocks, but also extend to non-resonant populations. A finite value of $T_0$ is therefore essential for suppressing these resonance-induced corrections.

\section{Rigorous Results}\label{sec:Results}

In this section, we present analytical results characterizing the performance of the algorithm. In Sec.~\ref{subsec:results}, we use the approximation of the channel $\mathcal{K}[\rho]$ by $\mathcal{K}_{\mathrm{LS}}[\rho]$ and the existence of an approximate fixed point close to $\rho_\beta$ to obtain a fixed-point error bound. Furthermore, we provide upper bounds for the variance of observables in Sec.~\ref{subsec:variance}.

\subsection{Fixed-point error bound}\label{subsec:results}

Equipped with the observation that the channel $\mathcal{K}[\rho]$ effectively implements a Lindbladian evolution followed by a unitary evolution and that this combined evolution has an approximate fixed point close to $\rho_\beta$, it is possible to derive the main result of this work: a rigorous bound on the fixed-point error controlled by $J^2$.
\begin{theorem}\label{thrm:fixedpointerror}
    Let $\rho_{\rm{fix}}$ denote the fixed point of the channel $\mathcal{K}[\rho]$. Denote the parameters of $\mathcal{K}$ as the width of the filter function $\sigma$ [cf. Eq.~\eqref{eq:f}], the randomization timescale $T_0$ [cf. Eq.~\eqref{eq:px}], the total evolution time $T$, and the inverse temperature $\beta$. If $\sigma<T$ and $T_0<T$, then
    \begin{align}
    \begin{split}
    &\norm{\rho_{\beta}-\rho_{\rm{fix}}}_1 \leq 2\epsilon + t_{\mathrm{mix},J}(\epsilon) n_B e^{\frac{\beta^2}{4\sigma^2}} e^{-\frac{T^2}{2\sigma^2 + 4 T_0^2}}
    \\ &\qquad + J^2 t_{\mathrm{mix},J}(\epsilon)\,\mathrm{poly}\left(n_B,\frac{\beta}{\sigma},\frac{T}{\sigma},\frac{T}{T_0}\right)e^{\frac{5\beta^2}{2\sigma^2}}
    \end{split}
    \end{align}
\end{theorem}
The full expression for the fixed-point error (including the explicit form of the ${\rm poly}(\cdots)$ term) is given in Eq.~\eqref{eq:totalbound}. The second term is a tail error, originating from the channel approximation in Eq.~\eqref{eq:channel}. This term is exponentially suppressed as a function of $T/\sigma$ and can therefore be made small by (mildly) increasing the evolution time between resets.
The third term of this error bound scales quadratically in $J$, and can thus be controlled by tuning the coupling strength $J$.

To obtain this bound, the strategy is to use the approximate fixed point of the channel $\mathcal{K}_{\mathrm{LS}}$, constructed in Sec.~\ref{subsec:Approx}. Recall that $\mathcal{K}_{\mathrm{LS}}[\tilde \rho]= 
\tilde{\rho} + \mathcal{O}(J^4)$ and also that $\Tilde \rho$ is close to the thermal state.
Since, as discussed in Sec.~\ref{subsec:ClosenessLindblad}, $\mathcal K_\mathrm{LS}$ is close to the effective channel $\mathcal K$ implemented by our protocol, a fixed-point bound on $\norm{\rho_\beta-\rho_{\rm{fix}}}_1$ then follows from a triangle inequality and closeness of fixed points of close channels. In particular (see Lemma~\ref{lem:Fixedpoints}, App.~\ref{app:fixedpointsclose} for details),
\begin{align}
\begin{split}
\label{eq:boundmain}
&\norm{\rho_{\mathrm{fix}}-\rho_{\beta}}_1
\leq \norm{\rho_{\mathrm{fix}}-\tilde{\rho}}_1 + \norm{\tilde{\rho}-\rho_{\beta}}_1 \\
&\leq 2\epsilon
+ \frac{t_{\mathrm{mix}, J}(\epsilon)}{J^2} \left(\norm{\mathcal{K}-\mathcal{K}_{\mathrm{LS}}}_{1\rightarrow 1}
+ \norm{\mathcal{K}_{\mathrm{LS}}[\tilde{\rho}]-\tilde{\rho}}_1\right)
\\&\phantom{\leq}+\norm{\tilde{\rho}-\rho_{\beta}}_1 .
\end{split}
\end{align}
Rigorous bounds for $\norm{\mathcal{K}_{\mathrm{LS}}[\tilde{\rho}]-\tilde{\rho}}_1$ and $\norm{\mathcal{K}-\mathcal{K}_{\mathrm{LS}}}_{1\rightarrow 1}$ are obtained in App.~\ref{app:boundsremaining}.
Substituting these above yields Theorem~\ref{thrm:fixedpointerror}.

Note that the inequalities in Eq.~\eqref{eq:boundmain} hold for any density matrix
\(\tilde{\rho}\). The key insight of choosing the approximate fixed point
\(\tilde{\rho}=\rho_\beta+J^2\sigma\) is that in this case all terms are controlled by a factor of \(J^2\). Indeed,
$\norm{\tilde{\rho}-\rho_\beta}_1
= J^2 \norm{\sigma}_1$ (for a $\sigma$ with bounded norm as $J\to0$) and $\frac{t_{\mathrm{mix},J}(\epsilon)}{J^2}
\norm{\mathcal{K}_{\mathrm{LS}}[\tilde{\rho}]-\tilde{\rho}}_1
= \mathcal{O}\!\left(J^2 t_{\mathrm{mix},J}(\epsilon)\right)$.

Provided that $t_{\mathrm{mix},J}(\epsilon)=o(J^{-2})$ (or, equivalently, $J^2\,t_{\mathrm{mix},J}(\epsilon)\to 0$ as  $J\to 0$), this result shows that the fixed-point error of the protocol can be made arbitrarily small by decreasing the coupling strength. Note that $t_{\mathrm{mix},J}$ here is already the rescaled mixing time; see Eq.~\eqref{eq:scsaling} and the surrounding discussion.
Consequently, to achieve the target accuracy $\epsilon$, one expects the scaling 
\begin{align}
J^2=\Tilde{\mathcal{O}}\left(\frac{\epsilon}{t_{\mathrm{mix},J}}\right).
\end{align}
Here, we have suppressed the dependence on the other protocol parameters and $\Tilde{\mathcal{O}}$ suppresses logarithmic dependencies on its argument.
The total Hamiltonian evolution time required to prepare an approximation to the thermal state is then given by 
\begin{align}
T_{\rm{tot}}=\Tilde{\mathcal{O}}\left(\frac{T t_{\mathrm{mix},J}^2}{\epsilon}\right).
\end{align}

Note that this bound qualitatively improves upon previous results obtained in Refs.~\onlinecite{hahn2026efficientquantumthermalstate,ding2025endtoendefficientquantumthermal}, which estimated that the total evolution time of the algorithm scales as $\frac{1}{\varepsilon^2}$.

\subsection{Bounds on the variance}\label{subsec:variance}
In practice, the averaged channel $\mathcal{K}[\rho]$ is not implemented directly. Rather, at each step a random evolution time $x_i$ is sampled, giving rise to the channel $\mathcal{K}_{x_i}[\rho]$. After $M$ steps, this defines the random channel
\begin{align}
    \mathcal{E}_M[\rho]=\mathcal{K}_{x_M}\circ \mathcal{K}_{x_{M-1}}\circ \cdots \circ \mathcal{K}_{x_1}[\rho].
\end{align}
Expectation values of observables are thus obtained by averaging over repeated realizations of the random channels $\mathcal{E}_M$, which reproduces the effect of $M$ applications of the same averaged channel $\mathcal{K}[\rho]$,
\begin{align}
\begin{split}
{\mathbb E}_M \mathcal{E}_M[\rho]
&= {\mathbb E}_M \mathcal{K}_{x_M}\circ \mathcal{K}_{x_{M-1}}\circ \cdots \circ \mathcal{K}_{x_1}[\rho] \\
&= \mathcal{K}^M[\rho].
\end{split}
\end{align}

However, the randomness induces an additional source of variance. To make this explicit, we introduce the following notation. 
For an arbitrary initial state $\rho$, denote the expectation value of $O$ with respect to this random sequence as $\braket{O}_M=\Tr[O \mathcal{E}_M[\rho]]$.
The variance of the observable $O$ can then be expressed as
\begin{align}
\begin{split}
    \mathrm{Var}(O)&=
    \mathbb{E}_M\left[\braket{O^2}_M -\braket{O}_M^2\right]\\&+
    \left[\mathbb{E}_M\braket{O}_M^2 -\left( \mathbb{E}_M \braket{O}_M \right)^2
    \right].
\end{split}
\end{align}
The first term is the usual shot-noise contribution, averaged over sequences, while the second term is the (sequence-to-sequence) variance of the mean $\expval{O}_M$. While the first term remains finite even in the deterministic limit $T_0 \to 0$, the second contribution is an additional contribution to the variance arising from the randomness in the channel.

The following theorem provides an upper bound on the variance of observables under repeated application of the randomized channel.

\begin{theorem}\label{thrm:variance}
Let $\mathcal{E}_M$ denote the random channel obtained after $M$ applications of the protocol. Then, for any observable $O$, the variance satisfies
\begin{align}
\begin{split}
    \mathrm{Var}(O)&\leq  \norm{O}^2_\infty+ \\&C^2 T_0^2\frac{4}{\pi}\norm{H}^2_\infty \norm{O}^2_\infty \frac{1-e^{-\frac {2 M}{\tau_{\rm{mix}}}}}{1-e^{-\frac{2}{\tau_{\rm{mix}}}}}
\end{split}
\end{align}
\end{theorem}
The proof of this bound is provided in App.~\ref{app:boundvariance}. The constant \(C\) is determined by the contraction estimate for a primitive channel
\begin{align}
    \norm{\mathcal{K}^{\dagger (n)}[O-\Tr[\rho_{\mathrm{fix}}O]\mathbb{I}]}_\infty
    < C e^{-n/\tau_{\mathrm{mix}}}\norm{O}_\infty.
\end{align}
The prefactor \(C\) depends on structural properties of the channel, such as the size of any Jordan blocks~\cite{Szehr_2014}. In general, \(C\) can be upper-bounded by a factor proportional to the mixing time. In many relevant cases, however, \(C=\mathcal{O}(1)\); for example, this is true when the channel is contractive, as at high temperatures~\cite{bakshi2025dobrushinconditionquantummarkov}.
We provide a numerical analysis of the fluctuations induced by a randomized evolution in Sec.~\ref {subsec:Observables}.

\begin{figure*}[t!]
    \centering
    \includegraphics[width=\linewidth]{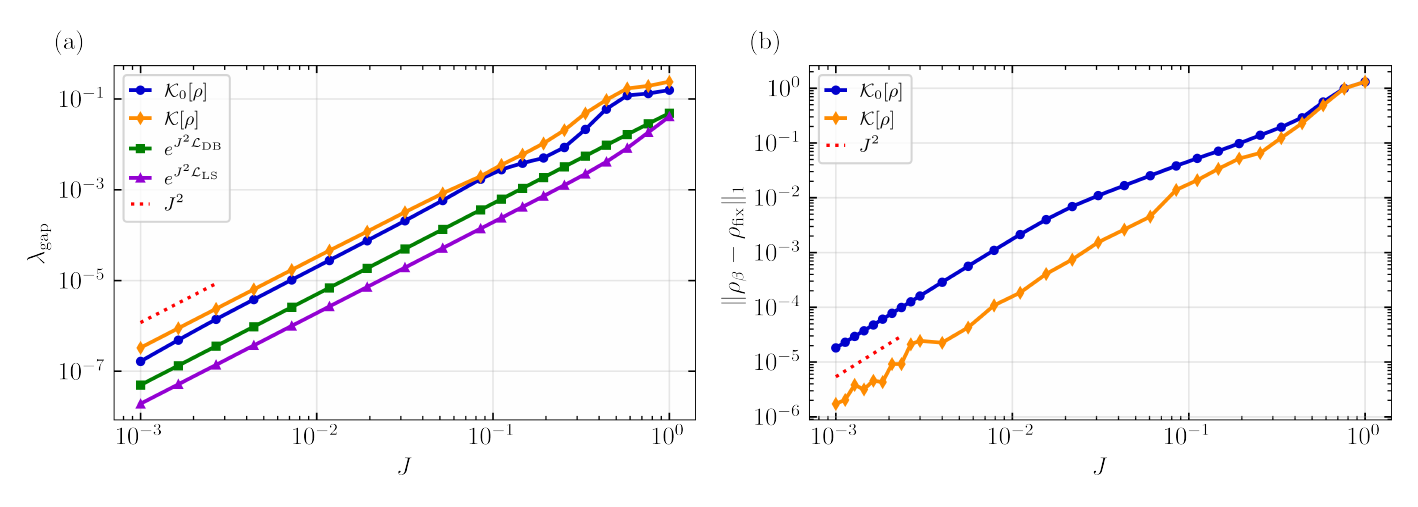}
    \caption{(a) Scaling of the spectral gap $\lambda_\mathrm{gap}$ as a function of the system-bath coupling strength $J$, for the channel $\mathcal{K}_0[\rho]$~(blue), $\mathcal{K}[\rho]$~(orange), $e^{J^2 \mathcal{L}_{\mathrm{LS}}}$~(purple) and $e^{J^2 \mathcal{L}_\mathrm{DB}}$~(green). The spectral gap scales in all cases as $J^2$.
    (b) Trace distance of fixed point from thermal state $\norm{\rho_\beta - \rho_\mathrm{fix}}_1$, for the channel $\mathcal{K}_0[\rho]$~(blue) and $\mathcal{K}[\rho]$~(orange). The fixed-point error scales as $J^2$ over a large range of coupling strengths. Both results shown for the mixed-field Ising model [Eq.~\eqref{eq:model}] with $n_S = 6,\ \beta = 1,\ T = 10$~(no resonance).}
    \label{fig:lambda_gap_J}
\end{figure*}

\section{Numerical analysis}\label{sec:Numerical results}

While Sec.~\ref{sec:Results} contains the main rigorous results, the present section complements them with additional numerical insights.
In Sec.~\ref{subsec:Spectral gap}, we analyze upper bounds for the mixing time of the channel in terms of the spectral gap, and provide numerical evidence for a quadratic scaling in $J$ over a wide range of coupling strengths.
In Sec.~\ref {subsec:Dephasing}, we investigate the role of the randomized evolution time used in the channel $\mathcal{K}[\rho]$~[cf. Eq.~\eqref{eq:channelrandomized}] and show that without it, the fixed-point error is non-perturbative in $J$.
 Finally, in Sec.~\ref{subsec:Observables}, we study the variance induced by different randomized implementations of the protocol.

Our analysis is based on the mixed-field Ising model
\begin{equation}\label{eq:model}
    H = \sum_i Z_i Z_{i+1} + g X_i + h Z_i.
\end{equation}
We choose $g = 0.9045,\ h = 0.809$, for which parameters the model is known to thermalize rapidly \cite{huse2013thermalisationparameters}. 
As jump operators, we choose the single qubit operator $A_a = X_a + Z_a$ on each qubit. Although other choices may yield faster mixing times, this choice is well-suited for studying the effect of resonances.

The width of the filter function $f(t)$ in Eq.~\eqref{eq:f} is specified by $\sigma = 1$. 
For all presented results, the mean evolution time satisfies $T > 4 \sigma$, and whenever the random evolution time is implemented, the width of $p(x)$ is set by $T_0 = 1$.
\subsection{Scaling of the spectral gap}\label{subsec:Spectral gap}

One important ingredient in the fixed-point error bound is the mixing time $\tau_{\rm mix}$. As shown in Sec.~\ref{subsec:ClosenessLindblad}, the channel $\mathcal{K}[\rho]$ implements up to a small error a Lindblad evolution of time $J^2$ followed by a unitary evolution of the system Hamiltonian. As long as the unitary evolution does not strongly affect the mixing, this suggests a scaling of the channel mixing time with $1/J^2$.

In this subsection, we numerically analyze an upper bound on the mixing time for the mixed-field Ising model [Eq.~\eqref{eq:model}] with $n_S=6$ sites and show that, in this case, this upper bound scales indeed as $1/J^2$.

While obtaining sharp estimates on the mixing time directly is challenging, we can analyze an upper bound in terms of the spectral gap~\cite{Temme_2010}. More specifically, let $\rho_{\rm fix}$ be the fixed point of a channel $\mathcal{K}[\rho]$, and consider the symmetrized channel
\begin{align}
    \mathcal{K}_{\rm sym}[\rho]
    =
    \rho_{\rm fix}^{-1/4}
    \mathcal{K}\!\left[\rho_{\rm fix}^{1/4} \rho \, \rho_{\rm fix}^{1/4}\right]
    \rho_{\rm fix}^{-1/4}.
\end{align}
Let $s_2$ denote the second-largest singular value of $\mathcal{K}_{\rm sym}[\rho]$, and define the spectral gap by
\begin{align}
    \lambda_{\rm gap} = 1 - s_2.
\end{align}
Then an upper bound on the mixing time is given by~\cite{Temme_2010}
\begin{align}
    \tau_{\mathcal{K},{\rm mix}}(\epsilon)
    \le
    \frac{1}{\lambda_{\rm gap}}
    \log\!\left(
        \frac{2 \norm{\rho_{\rm fix}^{-1/2}}}{\epsilon}
    \right).
\end{align}

In Fig.~\ref{fig:lambda_gap_J}~(a), we show the spectral gaps of $\mathcal{K}[\rho]$, $\mathcal{K}_0[\rho]$, $e^{J^2 \mathcal L_\mathrm{DB}}$ [Eq.~\eqref{eq:Gibbs}], and $e^{J^2 \mathcal L_\mathrm{LS}}$ [Eq.~\eqref{eq:LLSdefinition}].
The spectral gap scales as $J^2$ in all cases, which is consistent with a scaling of the mixing time $\tau_\mathrm{mix}$ with $1/J^2$ [Eq.~\eqref{eq:scsaling}].
As an interesting observation, including the unitary time evolution generated by the propagator $U_0$ appears to increase the spectral gap of the channel and therefore improves the corresponding upper bound on the mixing time. Understanding the origin of this effect remains an interesting question for future work.

In Fig.~\ref{fig:lambda_gap_J}~(b), we show the distance of the fixed points of $\mathcal{K}[\rho]$ and $\mathcal{K}_0[\rho]$ from the thermal state. This error scales with $J^2$, which matches our fixed-point error bounds [Theorem~\ref{thrm:fixedpointerror}].

\begin{figure*}[t!]
    \centering
    \includegraphics[width=\linewidth]{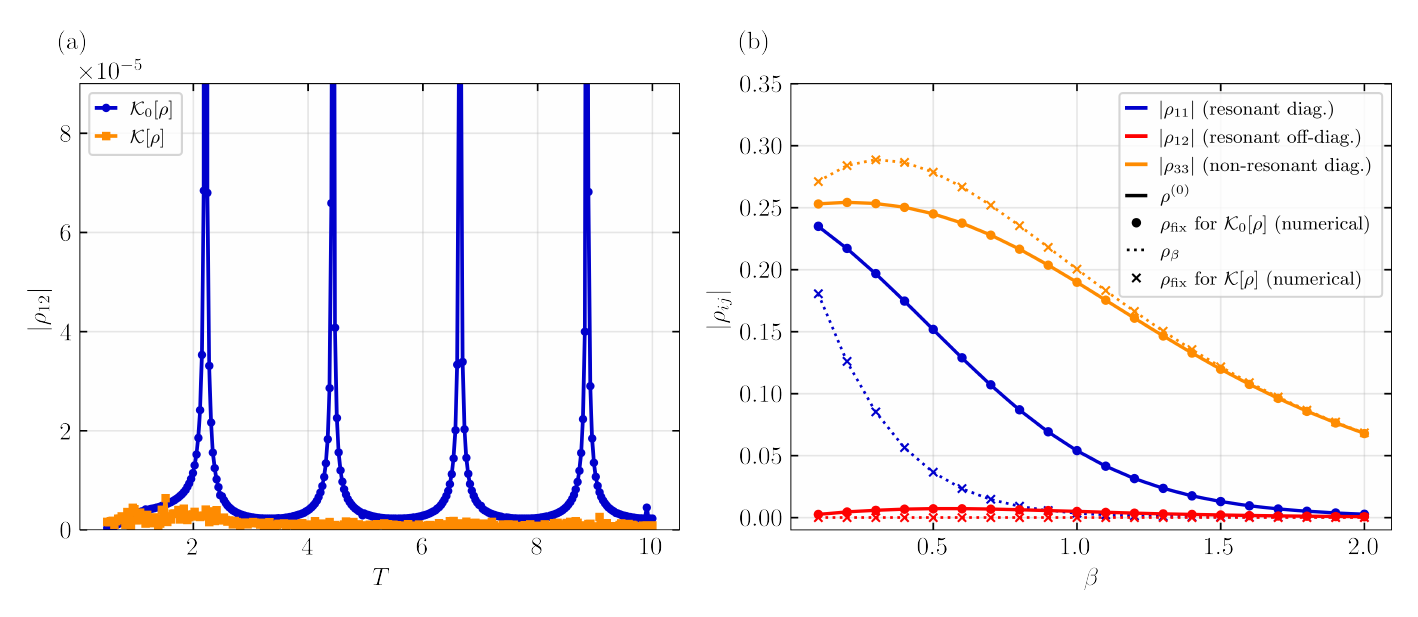}
    \caption{(a) Effect of the randomized evolution time [Eq.~\eqref{eq:channelrandomized}] in suppressing resonances. 
    Off-diagonal matrix element $|\rho_{12}|$ for the fixed point of $\mathcal{K}_0[\rho]$ [Eq.~\eqref{eq:single instance}]~(blue) and $\mathcal{K}[\rho]$ [Eq.~\eqref{eq:channelrandomized}]~(orange) as a function of the evolution time $T$. Resonance peaks are observed for $\mathcal{K}_0[\rho]$ but not for $\mathcal{K}[\rho]$ at $\omega_{12}T=2 \pi k$. Results shown for the mixed-field Ising model [Eq.~\eqref{eq:model}] with $n_S = 2,\ J = 0.01,\ \beta = 1$.
    (b) Effect of a resonance on the fixed point. Results shown for $n_S = 2,\ J = 0.01,\ T \approx 8.86$ to excite resonances between energy eigenstates 1 and 2. Diagonal $|\rho_{11}|$ in the resonant subspace~(blue), off-diagonal $|\rho_{12}|$ in the resonant subspace~(red), and diagonal $|\rho_{33}|$ in the non-resonant subspace~(orange) are compared for varying $\beta$. Resonances have no effect on the fixed points of $\mathcal{K}[\rho]$ ($\times$), which remains close to the thermal state $\rho_\beta$~(dotted). However, resonances cause the fixed points of $\mathcal{K}_0[\rho]$~($\bullet$) to deviate considerably, both in resonant and non-resonant blocks. This is correctly predicted by the solution for $\rho^{(0)}$ [Eq.~\eqref{eq:degeneratePTresult}]~(solid).}
    \label{fig:resonance_sweep_T_and_beta}
\end{figure*}

\subsection{Effect of random evolution times and resonances}\label{subsec:Dephasing}
In this subsection, we explore the effect of the randomized evolution time in the channel $\mathcal{K}[\rho]$ applied on the mixed-field Ising model [Eq.~\eqref{eq:model}] with $n_S = 2$ sites. As discussed in Sec.~\ref{subsec:Effectresonances}, many-body resonances can affect the fixed point significantly, and the randomization suppresses these effects.

As a first step, we directly probe the off-diagonal elements of the fixed point (in the energy eigenbasis) as a function of the evolution time $T$. We compare the results for the averaged channel $\mathcal{K}[\rho]$ with the deterministic channel $\mathcal{K}_{0}[\rho]$ corresponding to $x=0$ in Eq.~\eqref{eq:single instance}.  
The results are shown in Fig. \ref{fig:resonance_sweep_T_and_beta}~(a).
Resonance peaks occur for the channel $\mathcal{K}_0[\rho]$ at regular intervals in $T$ whenever $\omega_{ab}T=2 \pi k$ with $k$ an integer. In contrast, these resonances are suppressed by the channel $\mathcal{K}[\rho]$.

Next, we illustrate that resonances affect both resonant and non-resonant blocks of the fixed point. Choosing the evolution time $T$ to induce a resonance between two energy levels, the matrix elements of the fixed points of channels $\mathcal{K}_0[\rho]$ and $\mathcal{K}[\rho]$ are shown in Fig.~\ref{fig:resonance_sweep_T_and_beta}~(b) as a function of $\beta$. For the channel $\mathcal{K}[\rho]$, resonances have no visible effect, and the channel correctly converges to $\rho_\beta$ as $J = 0.01$ is sufficiently small. As for the channel $\mathcal{K}_0[\rho]$, many matrix elements in both resonant and non-resonant blocks significantly deviate from $\rho_\beta$, particularly at high temperatures. These results agree with the previous solution for $\rho^{(0)}$ [Eq.~\eqref{eq:degeneratePTresult}].

The effects of resonances, therefore, are not isolated to specific populations of the fixed point. The fixed point is entirely affected, and consequently, this leads to substantial fixed-point error. Fig. \ref{fig:trace_dist_resonance} shows the distance between the fixed point and $\rho_\beta$ as a function of $\beta$ when $T$ is chosen such that it  excites a resonance. Without the randomization step, $\mathcal{K}_0[\rho]$ has a large fixed-point error that does not vanish even in the $J\to 0^+$ limit, whereas the fixed-point error of $\mathcal{K}[\rho]$ is tunable by $J$.

\begin{figure}
    \centering
    \includegraphics[width=0.5\textwidth]{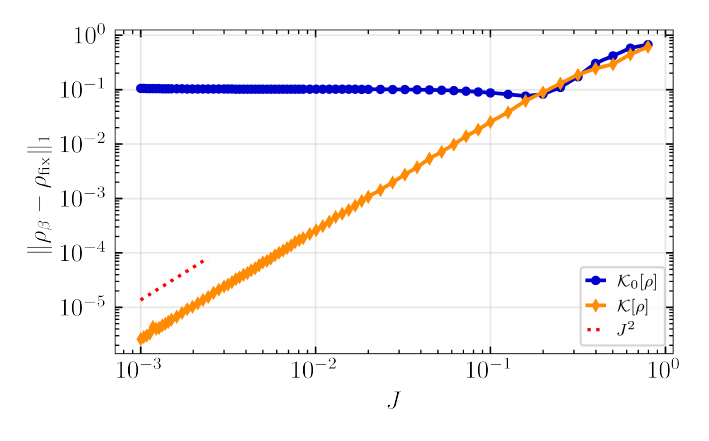}
    \caption{Trace distance of fixed point from thermal state $\norm{\rho_\beta - \rho_\mathrm{fix}}_1$ at a resonance, for the channel $\mathcal{K}_0[\rho]$~(blue) and $\mathcal{K}[\rho]$~(orange). Results shown for the mixed-field Ising model [Eq.~\eqref{eq:model}] with $n_S = 2,\ \beta = 1,\ T \approx 8.86$ to excite resonances between two energy eigenstates. Resonances lead to an $\mathcal{O}(J^0)$ fixed-point error for $\mathcal{K}_0[\rho]$, whereas the fixed-point error for $\mathcal{K}[\rho]$ remains $\mathcal{O}(J^2)$.}
    \label{fig:trace_dist_resonance}
\end{figure}

\subsection{Variance of observables}\label{subsec:Observables}
\begin{figure*}[t!]
    \centering
    \includegraphics[width=\linewidth]{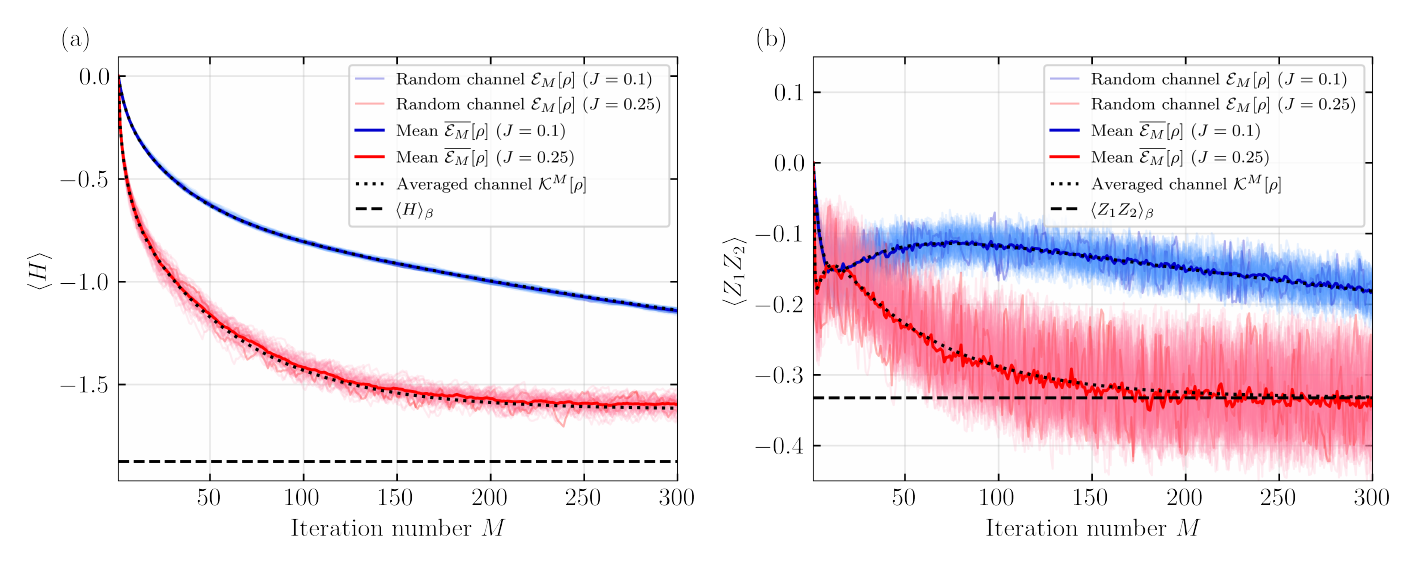}
    \caption{Evolution of thermal state preparation under 50 randomized channels $\mathcal{E}_M[\rho]$ (colored) and averaged channels $\mathcal{K}^M[\rho]$ (dotted). Results shown for the mixed-field Ising model [Eq.~\eqref{eq:model}] with $n_S = 2,\ \beta = 1,\ J = 0.1$~(blue) or 0.25~(red), $T = 10$ (no resonance), initial state $\rho \propto \mathbb{I}$ (maximally mixed state). (a) Evolution of average energy $\langle H \rangle$. (b) Evolution of two-point correlator $\langle Z_1 Z_2 \rangle$. The spread of results reflects the variance induced by random sampling. In contrast, averaging over random sequences significantly reduces the error and leads to a standard deviation (a) $\sigma_{H} = 0.008 \ (J=0.1),\ 0.034\ (J=0.25)$, (b) $\sigma_{Z_1 Z_2} = 0.019\ (J=0.1),\ 0.051\ (J=0.25)$. The equilibrium values of $\langle H \rangle, \langle Z_1 Z_2 \rangle$ deviate from those of the thermal state $\langle H \rangle_\beta, \langle Z_1 Z_2 \rangle_\beta$ (dashed) as these $J$ values are large, leading to fixed-point error.}
    \label{fig:variance_2qubit}
\end{figure*}

Finally, we investigate how using the randomized channel $\mathcal{E}_M[\rho]$, rather than the mean channel $\mathcal{K}[\rho]$, leads to additional fluctuations in observables. To this end, we construct 50 realizations of random channel sequences $\mathcal{E}_M[\rho]$ and measure the expectation value $\braket{O} = \Tr[O\,\mathcal{E}_M[\rho]]$ for a given observable $O$. We emphasize that the filter function and jump operators used here have not been optimized for rapid convergence, and we therefore expect that the preparation rate could be substantially improved by tailoring them to the relevant transition frequencies.

The results for two qubits are shown in Fig.~\ref{fig:variance_2qubit}, for the observables $O=H$ and $O=Z_1Z_2$. They indicate that the randomized implementation exhibits substantial fluctuations around the mean-channel result. Consequently, obtaining accurate estimates for observables cannot generally be done using a single random channel, and instead requires averaging over a number of channel realizations in parallel. This introduces an additional sampling overhead depending on the desired accuracy.

It is interesting to observe that the fluctuations decrease with increasing system size, as shown for $n_S=6$ sites in Fig.~\ref{fig:variance_6qubit}. One possible explanation in this case is that the mixed-field Ising model is thermalizing and therefore random trajectories concentrate. Understanding more precisely how the randomized evolution affects the variance would be a useful direction for future study.

As a final note, in addition to randomizing evolution times, other works \cite{ding2025endtoendefficientquantumthermal,chen2026overcominglambshiftsystembath} have also proposed randomizing the bath Hamiltonian and randomizing jump operators. While these are useful for bounding the mixing time, there is the practical trade-off of an increased variance. We present numerical evidence of this in App. \ref{app:morerandom}.

\begin{figure}
    \centering
    \includegraphics[width=0.5\textwidth]{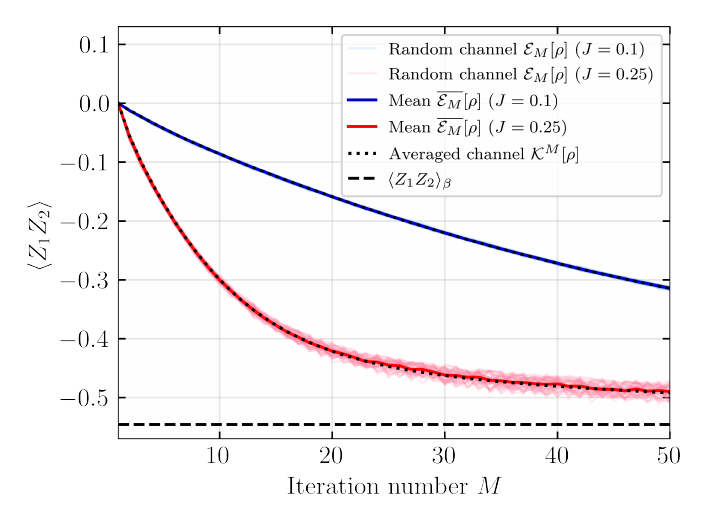}
    \caption{Evolution of two-point correlator $\langle Z_1 Z_2 \rangle$ under 50 randomized channels $\mathcal{E}_M[\rho]$~(colored) and averaged channels $\mathcal{K}^M[\rho]$~(dotted). Results shown for the mixed-field Ising model [Eq.~\eqref{eq:model}] with $n_S = 6,\ \beta = 1,\ J = 0.1$~(blue) or 0.25~(red), $T = 10$ (no resonance), initial state $\rho \propto \mathbb{I}$ (maximally mixed state). The final standard deviation is measured to be $\sigma_{Z_1 Z_2} = 0.002\ (J=0.1),\ 0.011\ (J=0.25)$, which is smaller than the 2-qubit case [Fig. \ref{fig:variance_2qubit}~(b)].}
    \label{fig:variance_6qubit}
\end{figure}

\section{Conclusion}\label{ref:Discussion}

In this work, we have studied an algorithm for thermal state preparation based on a time-dependent system-bath interaction. The algorithm requires only Hamiltonian evolution with local control over time-dependent interactions together with the ability to reset ancillas. It is therefore readily implementable on digital quantum devices and may also be feasible on analog quantum platforms.

Our first main result is
a rigorous fixed-point error bound that is governed by the mixing time of the channel and scales quadratically with the coupling strength \(J\). As a consequence, the error can, in principle, be made arbitrarily small by reducing \(J\). This result improves upon previous work on similar protocols, where such fixed-point error bounds did not explicitly account for the dependence on the system-bath coupling strength \cite{hahn2026efficientquantumthermalstate,ding2025endtoendefficientquantumthermal}, or were based on non-rigorous perturbative arguments \cite{lloyd2506quantum}. 

This favorable scaling in the coupling strength \(J\) comes at the cost of randomizing the duration of the system-bath interaction, and our analysis shows that this step cannot, in general, be omitted without degrading the accuracy of the thermal state preparation. In practice, the randomization means that the protocol implements evolution with a stochastic series of channels. This induces an additional variance of observables measured during the protocol: our second main contribution is to rigorously bound this resulting statistical error.

There are several important directions for future work. To fully establish practical feasibility of system-bath interaction protocols such as the one studied in this work, it is imperative to develop a more careful understanding of how to optimize the protocol with respect to its many parameters, and how optimal choices depend on features (e.g. the phase) of the system Hamiltonian and of its thermal state.
Recently, Ref.~\onlinecite{wang2026lindbladdynamicsrigorousguarantees} also showed that instead of using the system-bath coupling $J$ as a control parameter, a controlled fixed-point error can also be achieved by tuning parameters of the filter function at finite $J$. A natural task is therefore to clarify the differences, and possibly the interplay, between these complementary limits of controlled thermal state preparation using system-bath interactions.

On a more technical level, control of the fixed-point error in our bounds relies on the assumption that the rescaled mixing time of the protocol scales as \(o(J^{-2})\) as $J\to 0$. Although this assumption seems physically sound, and is supported by numerical results on the mixed-field Ising model, it is clearly desirable to understand its general validity and, if possible, provide rigorous guarantees. We flag this as an important avenue for subsequent investigations.

More generally, it would be valuable to understand the relation between the mixing time of thermal state preparation protocols based on system-bath interactions and  
their `proximate' exact Gibbs samplers.  For the latter class of Lindbladians, there is an increasing understanding of mixing times in physically relevant settings~\cite{rouze2024efficient,rouze2024optimal,bakshi2025dobrushinconditionquantummarkov,tong2024poly_mixing,smid2025poly_mixing,placke2024slow_mixing,rakovszky2024bottleneck,garmarnik2024slow_mixing}.

{\it A priori}, it is not entirely clear how to translate these results to near-term protocols, but obtaining mixing time estimates for the latter will be important to judge their applicability in practice. Here, we note that the related question of quantifying whether a generic system-bath model, possibly upon some suitable coarse-graining, is  `close' to an exact sampler has received some recent attention ~\cite{ScandiThermalization}, and may make a convenient starting point for exploring these issues.

\paragraph*{Note added:} During the completion of this manuscript, we became aware of Ref.~\onlinecite{chen2026overcominglambshiftsystembath}, which studies fixed-point errors for algorithms that prepare thermal states through system-bath interactions. While the channels analyzed in this work differ slightly from the one studied here, the strategies used to derive fixed-point error bounds as a function of $J$ agree with our approach.

\begin{acknowledgments}
We thank Jerome Lloyd and Zhiyan Ding for useful discussions and engaging correspondence.
We acknowledge support from a Leverhulme Trust International Professorship grant (Award Number:
LIP-2020-014, for a Leverhulme-Peierls Fellowship at Oxford to DH and BP) and the EPSRC under Grant No. EP/X030881/1 (DH) and a UKRI Frontier Research Grant/Horizon Europe
Guarantee (EPSRC Grant No. EP/Z002419/1, S.A.P.).
\end{acknowledgments}

\begin{appendix}
\begin{widetext}

\section{Closeness of fixed points and mixing times}\label{app:fixedpointsclose}
In this appendix, we present two lemmas adapted from Ref.~\onlinecite{ding2025endtoendefficientquantumthermal}. Lemma 1 connects the closeness of fixed points to the mixing time and the distance between the corresponding channels. This result is used repeatedly in the derivation of the fixed-point error bounds.

\begin{lemma}\label{lem:Fixedpoints}
    Let $\mathcal{K}_1[\rho]$ and $\mathcal{K}_2[\rho]$ be two channels with fixed points $\rho_1$ and $\rho_2$ respectively. Let $\tau_{1,\mathrm{mix}}(\epsilon)$ be the mixing time of $\mathcal{K}_1[\rho]$. The trace distance between fixed points $\rho_1$ and $\rho_2$ is bounded by
    \begin{subequations}
    \begin{align}
        \norm{\rho_1 - \rho_2}_1 &\leq \epsilon + \tau_{1,\mathrm{mix}}(\epsilon) \norm{\mathcal{K}_1[\rho_2] - \rho_2}_1 \label{theorem:mixingtime1} \\
        \norm{\rho_1 - \rho_2}_1 &\leq \epsilon + \tau_{1,\mathrm{mix}}(\epsilon) \norm{\mathcal{K}_1 - \mathcal{K}_2}_{1\to1} \label{theorem:mixingtime2}
    \end{align}
    \end{subequations}
\end{lemma}

\noindent \textit{Proof of Lemma 1.} The proof of Eq.~\eqref{theorem:mixingtime1} begins with the following split:
\begin{align}
    \norm{\rho_1 - \rho_2}_1 & \leq \norm{\rho_1 - \mathcal{K}_1^{\tau_{1,\mathrm{mix}}(\epsilon)}[\rho_2]}_1 + \norm{\mathcal{K}_1^{\tau_{1,\mathrm{mix}}(\epsilon)}[\rho_2] - \rho_2}_1.
\end{align}
The first term is bounded by $\epsilon$ by definition of $\tau_\mathrm{mix}(\epsilon)$ [Eq.~\eqref{eq:def_tau_mix}]. The second term can be expressed as a telescoping sum that is  simplified using the data processing inequality.
\begin{align}\begin{split}
    \norm{\mathcal{K}_1^{\tau_{1,\mathrm{mix}}(\epsilon)}[\rho_2] - \rho_2}_1 &\leq \sum_{n=0}^{\tau_{1,\mathrm{mix}}(\epsilon) - 1} \norm{\mathcal{K}_1^{n+1}[\rho_2] - \mathcal{K}_1^n[\rho_2]}_1 \\
    &\leq \sum_{n=0}^{\tau_{1,\mathrm{mix}}(\epsilon) - 1} \norm{\mathcal{K}_1[\rho_2] - \rho_2}_1 \\
    &= \tau_{1,\mathrm{mix}}(\epsilon) \norm{\mathcal{K}_1[\rho_2] - \rho_2}_1
\end{split}\end{align}

This completes the proof. The proof of Eq.~\eqref{theorem:mixingtime2} then follows immediately using the fixed-point relation $\rho_2 = \mathcal{K}_2[\rho_2]$,
\begin{align}
    \norm{\rho_1 - \rho_2}_1 \leq \epsilon + \tau_{1,\mathrm{mix}}(\epsilon) \norm{\mathcal{K}_1[\rho_2] - \mathcal{K}_2[\rho_2]}_1 \leq \epsilon + \tau_{1,\mathrm{mix}}(\epsilon) \norm{\mathcal{K}_1 - \mathcal{K}_2}_{1\to1}
\end{align}

Lemma 2 shows that two sufficiently close channels have comparable mixing times. This is used to justify that the exact channel $\mathcal{K}[\rho]$ has similar mixing times to the approximate channel $\mathcal{K}_{\mathrm{LS}}[\rho]$.

\begin{lemma}\label{lem:closemixingtimes}
    Let $\mathcal{K}_1[\rho]$ and $\mathcal{K}_2[\rho]$ be two channels with mixing times $\tau_{1,\mathrm{mix}}(\epsilon), \tau_{2,\mathrm{mix}}(\epsilon)$ respectively. If the two channels are sufficiently close with $\tau_{1,\mathrm{mix}}(\epsilon/2) \norm{\mathcal{K}_1 - \mathcal{K}_2}_{1\to1} \leq \epsilon/2$,
    \begin{equation}
        \tau_{2,\mathrm{mix}}(2\epsilon) \leq \tau_{1,\mathrm{mix}}(\epsilon/2)
    \end{equation}
\end{lemma}

\noindent \textit{Proof of Lemma 2.} Consider $\norm{\mathcal{K}_2^{\tau_{1,\mathrm{mix}}(\epsilon/2)}[\rho] - \rho_2}_1$ for any $\rho$, which can be split as follows.
\begin{equation}
    \norm{\mathcal{K}_2^{\tau_{1,\mathrm{mix}}(\epsilon/2)}[\rho] - \rho_2}_1 \leq \norm{\mathcal{K}_2^{\tau_{1,\mathrm{mix}}(\epsilon/2)}[\rho] - \mathcal{K}_1^{\tau_{1,\mathrm{mix}}(\epsilon/2)}[\rho]}_1 + \norm{\mathcal{K}_1^{\tau_{1,\mathrm{mix}}(\epsilon/2)}[\rho] - \rho_1}_1 +\norm{\rho_1 - \rho_2}_1
\end{equation}

The second term is bounded by $\epsilon/2$ by definition of $\tau_{1,\mathrm{mix}}(\epsilon/2)$ [Eq.~\eqref{eq:def_tau_mix}]. The third term is bounded by $\epsilon/2 + \tau_{1,\mathrm{mix}}(\epsilon/2) \norm{\mathcal{K}_1 - \mathcal{K}_2}_{1\to1}$ [Eq.~\eqref{theorem:mixingtime2}]. The first term requires a telescoping sum as follows,
\begin{align}\begin{split}
    \norm{\mathcal{K}_2^{\tau_{1,\mathrm{mix}}(\epsilon/2)}[\rho] - \mathcal{K}_1^{\tau_{1,\mathrm{mix}}(\epsilon/2)}[\rho]}_1 &\leq \sum_{n=0}^{\tau_{1,\mathrm{mix}}(\epsilon/2)-1} \norm{\mathcal{K}_2^{n+1} \mathcal{K}_1^{\tau_{1,\mathrm{mix}}(\epsilon/2)-n-1}[\rho] - \mathcal{K}_2^n \mathcal{K}_1^{\tau_{1,\mathrm{mix}}(\epsilon/2)-n}[\rho]}_1 \\
    &= \sum_{n=0}^{\tau_{1,\mathrm{mix}}(\epsilon/2)-1} \norm{\mathcal{K}_2^n \!\left( (\mathcal{K}_2-\mathcal{K}_1)\!\left[\mathcal{K}_1^{\tau_{1,\mathrm{mix}}(\epsilon/2)-n-1}[\rho]\right] \right)}_1 \\
    &\leq \sum_{n=0}^{\tau_{1,\mathrm{mix}}(\epsilon/2)-1} \norm{\mathcal{K}_2-\mathcal{K}_1}_{1\to 1} = \tau_{1,\mathrm{mix}}(\epsilon/2) \, \norm{\mathcal{K}_2-\mathcal{K}_1}_{1\to 1}.
\end{split}\end{align}

Collecting these results, the overall bound is given below, where the final inequality uses $\tau_{1,\mathrm{mix}}(\epsilon/2) \norm{\mathcal{K}_1 - \mathcal{K}_2}_{1\to1} \leq \epsilon/2$,
\begin{equation}
    \norm{\mathcal{K}_2^{\tau_{1,\mathrm{mix}}(\epsilon/2)}[\rho] - \rho_2}_1 \leq \epsilon + 2\tau_{1,\mathrm{mix}}(\epsilon/2) \, \norm{\mathcal{K}_2-\mathcal{K}_1}_{1\to 1} \leq 2\epsilon
\end{equation}

This is the definition of $\tau_{2,\mathrm{mix}}(2\epsilon)$. Thus it is proven that $\tau_{2,\mathrm{mix}}(2\epsilon) \leq \tau_{1,\mathrm{mix}}(\epsilon/2)$.

\section{Derivation of the fixed-point error bound}\label{sec:Derivationbigpicture}
In this section, we provide a detailed explanation of how a trace distance bound between the fixed point of $\mathcal{K}[\rho]$, $\rho_{\rm{fix}}$, and the Gibbs state $\rho_{\beta}$ can be derived.
While the overall strategy to bound the fixed-point error is straightforward, some of the calculations are a bit more involved and deferred to Appendix~\ref{app:boundsremaining}.
\subsection{Approximation of the channel}
As a first step, the channel $\mathcal{K}[\rho]$ is approximated by another channel that is easier to analyze (Sec.~\ref{subsec:ClosenessLindblad}). 
Denote the propagator of the system dynamics as $U_0(t) = e^{-iHt}$. 
Define
\begin{align}
    \mathcal{K}_{\mathrm{LS}}[\rho]= \int_{-\infty}^{\infty} \rd x\, p(x) U_0(T+x)\,\e^{J^2 \mathcal L_\mathrm{LS}}\rho \, U_0^\dagger(T+x).
\end{align}
The Lindbladian $\mathcal L_\mathrm{LS}[\rho]$ agrees up to a coherent term with the detailed balance Lindbladian $\mathcal L_\mathrm{DB}[\rho]$,
\begin{align}
    \mathcal L_\mathrm{LS}[\rho]=-i[\Delta G,\rho]+\mathcal L_\mathrm{DB}[\rho],
\end{align}
with
\begin{align}
    \Delta G=G_{\mathrm{LS}}-G_{\rm{DB}}.
\end{align}
The Lamb shift term $G_{\mathrm{LS}}$ is given by
\begin{align}
    G_{\mathrm{LS}} = -\frac{1}{2i} \sum_{a\in\mathcal A} \int_{-\infty}^{\infty} \rd t_1 \int_{-\infty}^{\infty} \rd t_2 \, \text{sign}(t_1 - t_2) f^*(t_2) f(t_1) A^\dagger_{a}(t_2) A_{a}(t_1). 
\end{align}

Though $\mathcal{K}_{\mathrm{LS}}[\rho]$ is a good approximation of the exact channel $\mathcal{K}[\rho]$, they differ at order $J^2$. Using a Dyson expansion in the interaction picture, it can be shown that [cf. App. \ref{app:boundchanneldist}]
\begin{align}\label{eq:channeldistbound}
    \norm{\mathcal{K}-\mathcal{K}_{\mathrm{LS}}}_{1\rightarrow 1} = \mathcal{O} \left( J^2 n_B e^{\frac{\beta^2}{4\sigma^2}} e^{-\frac{T^2}{2\sigma^2 + 4 T_0^2}} + J^4 n_B^2 e^{\frac{\beta^2}{2\sigma^2}} \right).
\end{align}
We denote the fixed point of $\mathcal{K}_{\mathrm{LS}}[\rho]$ by $\rho_\mathrm{LS}$ and the mixing time of the channel $\mathcal{K}_{\mathrm{LS}}[\rho]$ by $\tau_{\rm{mix}}\equiv\frac{t_{\mathrm{mix},J}}{J^2}$. 
It follows then from Lemma~\ref{lem:Fixedpoints}:
\begin{align}\label{eq:boundb2}
    \norm{\rho_\mathrm{fix}-\rho_\mathrm{LS}}_1\leq \epsilon+\frac{t_{\mathrm{mix},J}(\epsilon)}{J^2} \norm{\mathcal{K}-\mathcal{K}_{\mathrm{LS}}}_{1\rightarrow 1} \leq \epsilon+t_{\mathrm{mix},J}(\epsilon) \mathcal{O} \left( n_B e^{\frac{\beta^2}{4\sigma^2}} e^{-\frac{T^2}{2\sigma^2 + 4 T_0^2}} + J^2 n_B^2 e^{\frac{\beta^2}{2\sigma^2}} \right).
\end{align}
\subsection{Approximate fixed point of $\mathcal{K}_{\mathrm{LS}}$}
To bound $\norm{\rho_{\beta}-\rho_\mathrm{LS}}_1$, we construct an approximate fixed point of $\mathcal{K}_{\mathrm{LS}}$,
$\mathcal{K}_{\mathrm{LS}}[\tilde{\rho}]-\tilde{\rho}=\mathcal{O}(J^4)$ and split the bound using the triangle inequality
\begin{align} \label{eq:fixedpointdistbound}
    \norm{\rho_{\beta}-\rho_\mathrm{LS}}_1\leq \norm{\rho_{\beta}-\tilde{\rho}}_1+\norm{\tilde{\rho}-\rho_\mathrm{LS}}_1\leq \epsilon + \norm{\rho_{\beta}-\tilde{\rho}}_1 + \frac{t_{\mathrm{mix},J}(\epsilon)}{J^2} \norm{\tilde{\rho}-\mathcal{K}_{\mathrm{LS}}[\tilde{\rho}]}_1,
\end{align}
where the last inequality makes use of Lemma~\ref{lem:Fixedpoints}.

An ansatz for the approximate fixed point $\tilde \rho$ is obtained using degenerate perturbation theory, as outlined in Ref.~\onlinecite{lloyd2506quantum} and Sec.~\ref{subsec:Approx}.
While degenerate perturbation theory is not controlled, the obtained expression at first order in $J^2$ is already sufficient for our purposes.

To obtain an approximate fixed point, we make an ansatz $\tilde{\rho}=\rho^{(0)}+J^2\sigma$, expand the channel $\mathcal{K}_{\mathrm{LS}}[\rho]$ to first order in $J^2$ and match the resulting expression at each order in $J^2$.

All expressions are obtained in an eigenbasis of the Hamiltonian $H$: $H\ket{\phi_a}=E_a \ket{\phi_a}$, $A_{ab}=\braket{\phi_a|A|\phi_b}$ and $\omega_{ab}=E_a-E_b$.

To simplify the expression, we make use of the relation
\begin{align}
    \left[ \int_{-\infty}^{\infty} \rd x\, p(x) \,U_0\left(T+x\right)\,\rho \, U_0^\dagger\left(T+x\right) \right]_{ab} =e^{-\frac{1}{4}\omega_{ab}^2T_0^2-i\omega_{ab}T} \rho_{ab}.
\end{align}
Matching terms of $\mathcal{K}_{\mathrm{LS}}[\tilde \rho] = \tilde \rho$ to order $J^0$ gives
\begin{align}\label{eq:appendixorderJ0}
    e^{-\frac{1}{4}\omega_{ab}^2T_0^2-i\omega_{ab}T} \rho^{(0)}_{ab}=\rho^{(0)}_{ab},
\end{align}
and implies that ${\rho}^{(0)}$ is diagonal in the eigenbasis of $H$ for $T_0\neq 0$. 

The degeneracy is lifted by the next order $J^2$:
The relation for the diagonal terms gives
\begin{align}
(\mathcal L_\mathrm{LS}[\rho^{(0)}])_{aa}=0,
\end{align}
and it follows that $\rho^{(0)}=\frac{e^{-\beta E_a}}{Z}=\rho_\beta$.
The off-diagonal terms $a\neq b$ have to satisfy the relation
\begin{align}
   e^{-i\omega_{ab}T}\left(-i[\Delta G,\rho_\beta]_{ab}+e^{-\frac{1}{4}\omega_{ab}^2T_0^2} \sigma_{ab}\right)=\sigma_{ab},
\end{align}
This relation is satisfied by
\begin{align}\label{eq:defsigma}
    \sigma_{ab}=\frac{-i e^{-i\omega_{ab}T}[\Delta G,\rho_\beta]_{ab}}{1-e^{-\frac{1}{4}\omega_{ab}^2 T_0^2-i\omega_{ab}T}} = \frac{e^{-\beta {\varepsilon}_{ab}}}{Z} \frac{2i\,\sinh\left( \frac{\beta \omega_{ab}}{2} \right) e^{i\omega_{ab} T}}{e^{-\frac{1}{4}\omega_{ab}^2 T_0^2} - e^{i\omega_{ab}T}} (\Delta G)_{ab},
\end{align}
with $\varepsilon_{ab}=\frac{E_a+E_b}{2}$.
Insertion of the ansatz $\tilde{\rho}$ into $\mathcal{K}_{\mathrm{LS}}$ shows that indeed $\mathcal{K}_{\mathrm{LS}}[\tilde \rho]-\tilde \rho=\mathcal{O}(J^4)$.

The bounds for $\norm{\sigma}_1$ and $\norm{\tilde \rho-\mathcal{K}_{\mathrm{LS}}[\tilde \rho]}_1$ are a bit tedious, and are obtained in App.~\ref{app:boundsremaining}.
They are given by~[cf. App.~\ref{app:boundG}]
\begin{align}\label{eq:sigmabound}
    \norm{\rho_\beta-\tilde{\rho}}_1=J^2\norm{\sigma}_1 = \mathcal{O} \left( J^2 \frac{n_B \beta}{\sigma^3} \frac{T^6}{T_0^4} e^{\frac{9\beta^2}{4\sigma^2}} \right),
\end{align}
and~[cf. App.~\ref{app:boundlast}]
\begin{align}\label{eq:bound1}
    \norm{\tilde{\rho}-\mathcal{K}_{\mathrm{LS}}[\tilde{\rho}]}_1 = \mathcal{O} \left( J^4 n_B^2 \left( 1 + \frac{ \beta}{\sigma^3} \frac{T^6}{T_0^4} \right) e^{\frac{5\beta^2}{2\sigma^2}} \right).
\end{align}
This gives the bound:
\begin{align}\label{eq:bound1}
    \norm{\rho_{\beta}-\rho_\mathrm{LS}}_1\leq \epsilon + \mathcal{O} \left( J^2 \frac{n_B \beta}{\sigma^3} \frac{T^6}{T_0^4} e^{\frac{9\beta^2}{4\sigma^2}} \right) + t_{\mathrm{mix},J}(\epsilon) \mathcal{O} \left( J^2 n_B^2 \left( 1 + \frac{ \beta}{\sigma^3} \frac{T^6}{T_0^4} \right) e^{\frac{5\beta^2}{2\sigma^2}} \right).
\end{align}

\subsection{Combining the bounds}
Finally, combining Eq.~\eqref{eq:boundb2} and Eq.~\eqref{eq:bound1} gives the total fixed-point error:
\begin{align}\label{eq:totalbound}
\begin{split}
    \norm{\rho_{\beta}-\rho_{\rm{fix}}}_1 &\leq \norm{\rho_{\beta}-\rho_\mathrm{LS}}_1 + \norm{\rho_\mathrm{fix}-\rho_\mathrm{LS}}_1 \\
    &\leq 2\epsilon + \mathcal{O} \left( J^2 \frac{n_B \beta}{\sigma^3} \frac{T^6}{T_0^4} e^{\frac{9\beta^2}{4\sigma^2}} \right) + t_{\mathrm{mix},J}(\epsilon) \mathcal{O} \left( n_B e^{\frac{\beta^2}{4\sigma^2}} e^{-\frac{T^2}{2\sigma^2 + 4 T_0^2}} + J^2 n_B^2 e^{\frac{\beta^2}{2\sigma^2}} \right) \\
    &\qquad + t_{\mathrm{mix},J}(\epsilon) \mathcal{O} \left( J^2 n_B^2 \left( 1 + \frac{ \beta}{\sigma^3} \frac{T^6}{T_0^4} \right) e^{\frac{5\beta^2}{2\sigma^2}} \right).
\end{split}
\end{align}
Eq.~\eqref{eq:totalbound} is the explicit expression for the fixed-point error stated in Theorem~\ref{thrm:fixedpointerror}.

\section{Explicit calculations for the fixed-point error bound}\label{app:boundsremaining}

This appendix presents the detailed derivations underlying the fixed-point error bound summarized in App.~\ref{sec:Derivationbigpicture}. The derivation of the channel-approximation error \(\norm{\mathcal{K}-\mathcal{K}_{\mathrm{LS}}}_{1\rightarrow 1}\)~[cf. Eq.~\eqref{eq:boundb2}] is given in App.~\ref{app:boundchanneldist}. The remaining ingredients entering the bound on \(\norm{\rho_\beta-\rho_\mathrm{LS}}_1\) [Eq.~\eqref{eq:fixedpointdistbound}] are derived in the subsequent subsections: App.~\ref{app:boundG} treats the term \(\norm{\rho_\beta-\tilde{\rho}}_1\), while App.~\ref{app:boundlast} derives the bound on \(\norm{\tilde{\rho}-\mathcal{K}_{\mathrm{LS}}[\tilde{\rho}] }_1\).

\subsection{Bounds on the channel distance}\label{app:boundchanneldist}

In this subsection, the bound for the channel distance $\norm{\mathcal{K}-\mathcal{K}_{\mathrm{LS}}}_{1\rightarrow 1}$ stated in Eq.~\eqref{eq:channeldistbound} is derived.
To do so, observe:
\begin{align}
    \norm{\mathcal{K}-\mathcal{K}_{\mathrm{LS}}}_{1\rightarrow 1} \leq \int_{-\infty}^{\infty} \rd x\, p(x) \norm{ \mathcal{K}_x[\cdot] - e^{-iH(T+x)} e^{J^2 \mathcal L_\mathrm{LS}}[\cdot] e^{iH(T+x)} }_{1\rightarrow 1}.
 \end{align}
The channel $\mathcal{K}_x[\rho]$ can be represented as
\begin{align}
    \mathcal{K}_x[\rho]= e^{-iH(T+x)}\Tr_B \left[\tilde U_x [\rho \otimes \rho_B^0] \tilde U_x^\dagger \right] e^{iH(T+x)}\equiv e^{-iH(T+x)}\tilde{\mathcal{K}}_x[\rho] e^{iH(T+x)},
\end{align}
where $\tilde{U}_x$ is the propagator in the interaction picture
\begin{equation}
    \tilde U_x = \mathcal{T} \exp\left( -i \int_{-T/2}^{T/2+x} \rd t \, J V(t)\right).
\end{equation}

It follows from the unitary invariance of the norm,
\begin{align}\begin{split}\label{eq:boundKinfty}
    \norm{\mathcal{K}-\mathcal{K}_{\mathrm{LS}}}_{1\rightarrow 1} &\leq \int_{-\infty}^{\infty} \rd x\, p(x) \norm{e^{-iH(T+x)} \left[\tilde{\mathcal{K}}_x[\cdot]-e^{J^2 \mathcal L_\mathrm{LS}}[\cdot]\right] e^{iH(T+x)}}_{1 \rightarrow 1} \\
    &= \int_{-\infty}^{\infty} \rd x\, p(x) \norm{\tilde{\mathcal{K}}_x -e^{J^2 \mathcal L_\mathrm{LS}}}_{1\rightarrow 1},
\end{split}\end{align}
and thus it suffices to analyze the channels $\tilde{\mathcal{K}}_x[\rho]$ and $\e^{J^2 \mathcal L_\mathrm{LS}}[\rho]$.

Our strategy is to expand both channels in powers of $J$. For the channel $\Tilde{\mathcal{K}}_x$, we use a Dyson series for the system-bath propagator, and for the other channel we use the expansion of the exponential function. Keeping terms up to order $J^2$, this yields
\begin{align}
    \tilde{\mathcal{K}}_x[\rho] &= \rho + J^2 \tilde{\mathcal{K}}_x^{(1)}[\rho] + \mathcal{O}\left(J^4 n_B^2\norm{f}_{L^\infty}^4\right), \\
    e^{J^2 \mathcal L_\mathrm{LS}}[\rho] &= \rho + J^2 \mathcal L_\mathrm{LS}[\rho] + \mathcal{O}(J^4 \norm{\mathcal{L}_{\mathrm{LS}}}_{1\rightarrow 1}^2).
\end{align}
Comparing these two terms yields a bound on their difference

\begin{equation}\label{eq:channeldistboundsplit}
    \Vert \tilde{\mathcal{K}}_x - \e^{J^2 \mathcal L_\mathrm{LS}} \Vert_{1\rightarrow 1} \leq \Vert \e^{J^2 \mathcal L_\mathrm{LS}} - \tilde{\mathcal{K}}_x \Vert_{1\rightarrow 1}= J^2 \Vert \mathcal{L}_{\mathrm{LS}} - \tilde{\mathcal{K}}_x^{(1)} \Vert_{1\rightarrow 1} + \mathcal{O}\left(J^4 n_B^2(\norm{f}^4_{L^\infty}+\norm{\mathcal{L}_{\mathrm{LS}}}_{1\rightarrow 1}^2)\right).
\end{equation}

We will now use the explicit expressions for $\mathcal L_\mathrm{LS}$ and $\Tilde K_x^{(1)}$ to bound the difference in their image (the numerator of the rhs above) for arbitrary $\rho$.
To this end, we first derive the expression for $\tilde{\mathcal{K}}_x^{(1)}[\rho]$, using the Dyson expansion for $\tilde U_x$ given by
\begin{equation}\begin{split}
    \tilde U_x &= \mathbb{I} - i J \int_{-T/2}^{T/2+x} \rd t_1 \, V(t_1) - J^2 \int_{-T/2}^{T/2+x} \rd t_1 \int_{-T/2}^{t_1} \rd t_2 \, V(t_1) V(t_2) \\
    &+ i J^3 \int_{-T/2}^{T/2+x} \rd t_1 \int_{-T/2}^{t_1} \rd t_2 \int_{-T/2}^{t_2} \rd t_3 \, V(t_1) V(t_2) V(t_3) + \mathcal{O} \left(J^4 \left[\int_{-T/2}^{T/2+x} \rd t \, \norm{V(t)}_\infty \right]^4 \right).
\end{split}\end{equation}
Odd orders in $J$ vanish after tracing out the bath degrees of freedom.
At order $J^2$, there are three contributions:
\begin{align}\begin{split}
    &\text{Tr}_B \left[ \left(- i J \int_{-T/2}^{T/2+x} \rd t_1 \, V(t_1)\right) [\rho \otimes \rho_B^0] \left(i J \int_{-T/2}^{T/2+x} \rd t'_1 \, V^\dagger(t'_1)\right) \right] = J^2 \sum_{a} \tilde L_{a} \rho \tilde L_{a}^\dagger, \\
    &\text{Tr}_B \left[ \left(- J^2 \int_{-T/2}^{T/2+x} \rd t_1 \int_{-T/2}^{t_1} \rd t_2 \, V(t_1) V(t_2) \right)  [\rho \otimes \rho_B^0] \right] \\
    &\qquad = -J^2 \sum_a \int_{-T/2}^{T/2+x} \rd t_1 \int_{-T/2}^{t_1} \rd t_2 \, f^*(t_1) f(t_2) A^\dagger_{a}(t_1) A_{a}(t_2) \rho,\\
    &\text{Tr}_B \left[ [\rho \otimes \rho_B^0] \left(- J^2 \int_{-T/2}^{T/2+x} \rd t'_1 \int_{-T/2}^{t'_1} \rd t'_2 \, V^\dagger(t'_2) V^\dagger(t'_1) \right) \right] \\
    &\qquad = -J^2 \rho \sum_a \int_{-T/2}^{T/2+x} \rd t'_1 \int_{-T/2}^{t'_1} \rd t'_2 \, f(t'_1) f^*(t'_2) A^\dagger_{a}(t'_2) A_a(t'_1).
\end{split}\end{align}
It can be shown that these terms sum up to $J^2 \tilde{\mathcal{K}}^{(1)}_x[\rho]$, where
\begin{equation} \label{eq:Kx1}
    \tilde{\mathcal{K}}_x^{(1)}[\rho] = -i[\tilde G_{\mathrm{LS}},\rho] + \sum_a \tilde L_{a} \rho \tilde L_{a}^\dagger - \left\{ \frac{1}{2} \tilde L_{a}^\dagger \tilde L_{a}, \rho \right\}.
\end{equation}
Here, $\tilde G_{\mathrm{LS}}$ and $\tilde L_{a}$ are nearly identical to $G_{\mathrm{LS}}$ [Eq.~\eqref{eq:GLS}] and $L_a$ [Eq.~\eqref{eq:jump}], but with their integration domains restricted to $(-T/2,\, T/2+x)$ instead of $(-\infty, \infty)$. This causes $\tilde{\mathcal{K}}^{(1)}_x[\rho]$ to differ from $\mathcal{L}_{\mathrm{LS}}[\rho]$ by a finite-time tail error.

This completes the expansion of $\tilde{\mathcal{K}}_x[\rho]$:
\begin{align}
    \tilde{\mathcal{K}}_x[\rho]=\rho+J^2 \tilde{\mathcal{K}}_x^{(1)}[\rho]+ \mathcal{O} \left(n_B^2 \left[\int_{-T/2}^{T/2+x} \rd t \, |f(t)| \norm{A_a(t)}_\infty \right]^4 \norm{\rho}_1 \right)=\rho+J^2 \tilde{\mathcal{K}}_x^{(1)}[\rho]+ \mathcal{O} \left(J^4 n_B^2 e^{\frac{\beta^2}{2\sigma^2}} \right).
\end{align}

We continue with bounding $\Vert \mathcal{L}_{\mathrm{LS}}[\rho] - \tilde{\mathcal{K}}_x^{(1)}[\rho] \Vert_1$, which represents the tail error. We obtain from Eq.~\eqref{eq:Kx1},
\begin{equation}\begin{split}
    \mathcal{L}_{\mathrm{LS}}[\rho] - \tilde{\mathcal{K}}_x^{(1)}[\rho] &= \sum_a \iint\nolimits_{\mathbb{R}^2 \setminus (-T/2,\,T/2+x)} dt_1\, dt_2 \, f^*(t_1) f(t_2) \Bigg[ -\frac 12 \text{sign}(t_1 - t_2) \left[A_{a}^\dagger(t_1) A_{a}(t_2), \rho \right] \\
    &\qquad - \frac12  \{ A^\dagger_{a}(t_1) A_{a}(t_2), \rho\} + A_{a}(t_2) \rho A^\dagger_{a}(t_1)
    \Bigg].
\end{split}\end{equation}

Using $\Vert A_a(t) \Vert_\infty \leq 1$, Holder's inequality gives $\left \Vert \left[ A_{a}(t_1) A_{a}(t_2), \rho \right] \right \Vert_1 \leq 2 \Vert \rho \Vert_1$, $\left \Vert \{ A_{a}(t_1) A_{a}(t_2), \rho \} \right \Vert_1 \leq 2 \Vert \rho \Vert_1$, $\left\Vert A_{a}(t_2) \rho A^\dagger_{a}(t_1) \right\Vert_1 \leq \Vert \rho \Vert_1$. This leads to
\begin{equation}\begin{split}
    \Vert \mathcal{L}_{\mathrm{LS}}[\rho] - \tilde{\mathcal{K}}^{(1)}[\rho] \Vert_1 &\leq 3 n_B \Vert \rho \Vert_1 \iint\nolimits_{\mathbb{R}^2 \setminus (-T/2,\, T/2+x)} \rd t_1\, \rd t_2 \, |f(t_1)| |f(t_2)| \\
    &\leq 6 n_B \Vert \rho \Vert_1 \left[ \int_{-\infty}^{\infty} \rd t\, |f(t)| \right] \left[ \int_{|t|>T/2-|x|} \rd t \, |f(t)| \right].
\end{split}\end{equation}

These integrals can be evaluated for the form of $f(t)$ chosen in Eq.~\eqref{eq:f} 
and thus give
\begin{equation}
    \Vert \mathcal{L}_{\mathrm{LS}}[\rho] - \tilde{\mathcal{K}}^{(1)}[\rho] \Vert_1 < 6n_B  e^{\frac{\beta^2}{4\sigma^2}}\text{erfc} \left( \frac{T-2|x|}{\sqrt{2}\sigma} \right)\norm{\rho}_1. 
\end{equation}

To obtain the explicit expression for  $\mathcal{O}\left(J^4\Vert \mathcal{L}_{\mathrm{LS}}\Vert^2_{1\rightarrow1}\right)$, we use 
\begin{equation} \label{eq:LLSnormbound}
    \Vert \mathcal{L}_{\mathrm{LS}}[\rho] \Vert_1 < 3 n_B e^{\frac{\beta^2}{4\sigma^2}} \Vert \rho \Vert_1,
\end{equation}
and thus get
\begin{equation} \label{eq:LLS2normbound}
    \mathcal{O}\left(J^4\Vert \mathcal{L}_{\mathrm{LS}}\Vert^2_{1\rightarrow1}\right) = \mathcal{O}\left(J^4 n_B^2 e^{\frac{\beta^2}{2\sigma^2}} \right).
\end{equation}
Combining all terms, this gives
\begin{equation}
    \Vert \tilde{\mathcal{K}}_x - \tilde{\mathcal{K}}_{\infty} \Vert_{1\rightarrow 1} = \mathcal{O} \left( J^2 n_B \, \text{erfc} \left( \frac{T-2|x|}{\sqrt{2}\sigma} \right) e^{\frac{\beta^2}{4\sigma^2}} + J^4 n_B^2 e^{\frac{\beta^2}{2\sigma^2}} \right).
\end{equation}
The final integration over $p(x)$ in Eq.~\eqref{eq:boundKinfty} leads to the final bound 
\begin{align}
\Vert \mathcal{K} -  \mathcal{K}_{\mathrm{LS}} \Vert_{1\rightarrow 1}= \mathcal{O} \left( J^2 n_B e^{\frac{\beta^2}{4\sigma^2}} e^{-\frac{T^2}{2\sigma^2 + 4 T_0^2}} + J^4 n_B^2 e^{\frac{\beta^2}{2\sigma^2}} \right).
\end{align}
\subsection{Bounds on the approximate fixed point}\label{app:boundG}
In this subsection, the bound for the approximate fixed point $\norm{\rho_\beta-\tilde{\rho}}_1=J^2 \norm{\sigma}_1$ stated in Eq.~\eqref{eq:sigmabound} is derived, where $\tilde{\rho}$ is defined by $\mathcal{K}_{\mathrm{LS}}[\tilde{\rho}]-\tilde{\rho}=\mathcal{O}(J^4)$.

From the definition in Eq.~\eqref{eq:defsigma}, it follows that $\sigma$ can be expressed as
\begin{equation}
    \sigma = \rho_\beta^{1/2} \Omega \rho_\beta^{1/2},
\end{equation}
where $\Omega$ is a matrix whose components in the energy eigenbasis are given by
\begin{equation}
    \Omega_{ab} = F(\omega_{ab}) (\Delta G)_{ab},
\end{equation}
with
\begin{equation}
    F(\omega_{ab}) = \frac{2i\,\sinh\left( \frac{\beta \omega_{ab}}{2} \right) e^{i\omega_{ab} T}}{e^{-\frac{1}{4}\omega_{ab}^2 T_0^2} - e^{i\omega_{ab}T}}.
\end{equation}
Using Holder's inequality, and $\Vert \rho_\beta^{1/2} \Vert_2 = 1$, this gives 
\begin{equation}
    \Vert \sigma \Vert_1 \leq \Vert \rho_\beta^{1/2} \Vert_2 \Vert \Omega \Vert_\infty \Vert \rho_\beta^{1/2} \Vert_2 = \Vert \Omega \Vert_\infty.
\end{equation}
It remains to bound $\Vert \Omega \Vert_\infty$. 

To do this, split $\Omega = \Omega_\mathrm{LS} - \Omega_\mathrm{DB}$, and consider each of their Bohr frequency representations,
\begin{subequations}
\begin{align}
    \Omega_\mathrm{DB} &= \sum_{a} \sum_{\nu_1,\nu_2\in\mathfrak{B}} \frac{1}{2i} F(\nu_1-\nu_2 ) \tanh \left( \frac{\beta (\nu_1 - \nu_2)}{4} \right) \int_{-\infty}^\infty \rd t_1 \int_{-\infty}^\infty \rd t_2 \, f(t_1) f^*(t_2) e^{i\nu_1 t_1 - i\nu_2 t_2} (A_{\nu_2}^a)^\dagger A_{\nu_1}^a, \label{eq:omegaDB} \\
    \Omega_\mathrm{LS} &= \sum_{a} \sum_{\nu_1,\nu_2\in\mathfrak{B}} -\frac{1}{2i} F(\nu_1 - \nu_2) \int_{-\infty}^\infty \rd t_1 \int_{-\infty}^\infty \rd t_2 \, \text{sign}(t_1-t_2) f(t_1) f^*(t_2) e^{i\nu_1 t_1 - i\nu_2 t_2} (A_{\nu_2}^a)^\dagger A_{\nu_1}^a. \label{eq:omegaLS}
\end{align}
\end{subequations}
\subsubsection{Bounding $\Omega_\mathrm{DB}$}
To bound $\Vert \Omega_\mathrm{DB} \Vert_\infty$ [Eq.~\eqref{eq:omegaDB}], 
make a change of variables $t_+ = \frac{t_1+t_2}{2}$, $t_- = \frac{t_1-t_2}{2}$, $\nu_+ = \nu_1+\nu_2$, $\nu_- = \nu_1-\nu_2$ and substitute $f(t) = \sqrt{\frac{2}{\pi \sigma^2}} e^{-\frac{2}{\sigma^2}\left(t -\frac{i\beta}{4} \right)^2}$.
This gives a separable product $\hat h_+(\nu_+) \, \hat h_-(\nu_-)$,
\begin{equation}
    \Omega_\mathrm{DB} = \sum_a \sum_{\nu_1,\nu_2\in\mathfrak{B}} -\frac{2i}{\pi\sigma^2} \underbrace{F(\nu_- ) \tanh \left( \frac{\beta\nu_-}{4} \right) \int_{-\infty}^\infty \rd t_- \, e^{-i\nu_- t_-} e^{-\frac{4}{\sigma^2}t_-^2}}_{\hat h_-(\nu_-)} \underbrace{\int_{-\infty}^\infty \rd t_+ \, e^{-i\nu_+ t_+} e^{-\frac{4}{\sigma^2}\left( t_+ + \frac{i\beta}{4} \right)^2}}_{\hat h_+(\nu_+)} (A_{\nu_2}^a)^\dagger A_{\nu_1}^a.
\end{equation}
Corollary A.1 of Ref.~\onlinecite{CKG23} states that
\begin{equation}
    \sum_a \sum_{\nu_1,\nu_2\in\mathfrak{B}} \hat h_+(\nu_+) \hat h_-(\nu_-) (A_{\nu_2}^a)^\dagger A_{\nu_1}^a = \sum_a \int_{-\infty}^{\infty} \rd t_- \, h_-(t_-) e^{-iH t_-} \left( \int_{-\infty}^{\infty} \rd t_+ \, h_+(t_+) A_{a}^\dagger(t_+) A_{a}^\dagger(-t_+) \right) e^{iH t_-},
    \label{id_chenA1}
\end{equation}
where the Fourier transform is defined as $\hat h(\nu) = \int_{-\infty}^{\infty} \rd t \, h(t)e^{-i\nu t} \Leftrightarrow h(t) =  \frac{1}{2\pi} \int_{-\infty}^{\infty} \rd\nu \, h(\nu) e^{i\nu t}$.

Using $\Vert e^{iHt} \Vert_\infty = 1$, and $\Vert A_a(t) \Vert_\infty \leq 1$, a bound for $\Vert \Omega_\mathrm{DB} \Vert_\infty$ is therefore given by
\begin{equation}\begin{split}\label{eq_omegaDBinfnorm}
    \Vert \Omega_\mathrm{DB} \Vert_\infty &= \frac{2}{\pi \sigma^2} \left\Vert \sum_a \int_{-\infty}^{\infty} \rd t_- \, h_-(t_-) e^{-iH t_-} \left( \int_{-\infty}^{\infty} \rd t_+ \, h_+(t_+) A_{a}^\dagger(t_+) A_{a}^\dagger(-t_+) \right) e^{iH t_-} \right\Vert_\infty \\
    &\leq \frac{2n_B}{\pi \sigma^2} \int_{-\infty}^{\infty} \rd t_- \, |h_-(t_-)| \int_{-\infty}^{\infty} \rd t_+ \, |h_+(t_+)|.
\end{split}\end{equation}
$\int_{-\infty}^{\infty} \rd t_+ \, |h_+(t_+)|$ can be found explicitly through an inverse Fourier transform of $\hat h_+(\nu_+)$. This leads to
\begin{equation}
    \int_{-\infty}^\infty \rd t_+ \, |h_+(t_+)| = \frac{\sigma}{4\sqrt\pi} e^{\frac{\beta^2}{4\sigma^2}}.
    \label{eq_h+1normbound}
\end{equation}
It remains to estimate 
$\int_{-\infty}^{\infty} \rd t_- \, |h_-(t_-)|$ . To do so, consider the explicit expression for $\hat h_-(\nu)$,
\begin{equation}
    \hat h_-(\nu) = 2i \sqrt \pi \sigma \frac{e^{i\nu T}}{e^{-\frac{1}{4}\nu^2 T_0^2} -  e^{i\nu T}} \sinh^2\left( \frac{\beta \nu}{4} \right) e^{-\frac{\nu^2 \sigma^2}{16}}.
\end{equation}

Our strategy is to express $\hat h_-(\nu)$ as a product of two appropriately chosen functions $\hat h_-(\nu) = 2i\sqrt\pi \sigma \ \hat p(\nu) \hat q(\nu)$, so that by the convolution theorem, $h_-(t) = 2i\sqrt\pi \sigma \ (p * q)(t)$. $\int_{-\infty}^{\infty} \rd t \, |h_-(t)|$ can then be bounded using Young's convolution inequality,
\begin{equation}\label{eq_h-1normbound}
    \int_{-\infty}^{\infty} \rd t \, |h_-(t)| \equiv \Vert h_-(t) \Vert_{L^1} = 2 \sqrt{\pi} \sigma \Vert (p*q)(t) \Vert_{L^1} \leq 2 \sqrt{\pi} \sigma \Vert p(t) \Vert_{L^1} \Vert q(t) \Vert_{L^1}.
\end{equation}
Here, $p(t), q(t)$ denote the inverse Fourier transforms of $\hat{p}(\nu), \, \hat{q}(\nu)$.
Denoting $\alpha = \frac{\sigma^2}{32}$, we pick
\begin{equation}
    \hat p(\nu) = \frac{1}{\nu} \sinh^2 \left( \frac{\beta \nu}{4} \right) e^{i\nu T - \alpha \nu^2},\qquad
    \hat q(\nu) = \frac{\nu e^{-\alpha \nu^2}}{e^{-\frac{1}{4}\nu^2 T_0^2} - e^{i\nu T}}.
\end{equation}

To determine $\Vert p(t) \Vert_{L^1}$, the inverse Fourier transform $p(t)$ is expressed by substituting $\frac{1}{\nu} \sinh^2 \left( \frac{\beta \nu}{4} \right) = \int_0^{\beta/4} \rd x \, \sinh(2\nu x)$ and reversing the order of integration.
\begin{equation}
    p(t) = \frac{1}{2\pi}\int_0^{\beta/4} \rd x  \int_{-\infty}^{\infty} \rd \nu \, \sinh(2\nu x) e^{i\nu (T+t) - \alpha \nu^2} = \frac{i}{2\sqrt{\pi \alpha}} e^{-\frac{(T+t)^2}{4\alpha}} \int_0^{\beta/4} \rd x \, e^{\frac{x^2}{\alpha}} \sin \left( \frac{T+t}{\alpha}x \right).
\end{equation}
$|p(t)|$ is now bounded using $\left| \sin \left( \frac{T+t}{\alpha} \right) \right| < 1$ to obtain $|p(t)| \leq \frac{\beta}{8\sqrt{\pi \alpha}} e^{-\frac{(T+t)^2}{4\alpha}} e^{\frac{\beta^2}{16\alpha}}$. It follows that
\begin{equation}
    \norm{p(t)}_{L^1} \leq \frac{\beta}{4} e^{\frac{2\beta^2}{\sigma^2}}. 
\end{equation}

To find $\Vert q(t) \Vert_{L^1}$, we split the $t$-integral into short times $|t|<t_0$ and long times $|t|>t_0$, with $t_0$ to be determined later:
\begin{equation}
    \Vert q(t) \Vert_{L^1} = \frac{1}{2\pi} \int_{-t_0}^{t_0} \rd t \, \left| \int_{-\infty}^{\infty} \rd \nu \, e^{i\nu t} \hat q(\nu) \right| +  \frac{1}{2\pi} \int_{|t|>t_0} \rd t \, \left| \int_{-\infty}^{\infty} \rd \nu \, e^{i\nu t} \hat q(\nu) \right|.
\end{equation}
Using partial integration, we obtain
\begin{equation}
    \Vert q(t) \Vert_{L^1} = \frac{2t_0}{\pi} \int_{0}^{\infty} \rd \nu \,|\hat q(\nu)| + \frac{2}{\pi t_0} \int_{0}^{\infty} \rd \nu \, \left|\frac{\rd^2 \hat q(\nu)}{\rd\nu^2}\right|.
\end{equation}
To obtain a bound for $|\hat q(\nu)|$, note that $\frac{1}{\left| e^{-\frac{1}{4}\nu^2 T_0^2} - e^{i\nu T} \right|} = \Theta(\nu^{-1})$. For a bound of $\left|\frac{\rd^2 \hat q(\nu)}{\rd\nu^2}\right|$, the numerator can be bounded by a term of order $\mathcal{O}(\nu^3 e^{-\alpha \nu^2})$ and thus cancels divergences in the denominator.
A tedious but straightforward calculation then yields, with the assumption $T_0 < T$ and $\sigma < T$,
\begin{equation}
    \int_{0}^{\infty} \rd \nu \, |\hat q(\nu)| = \mathcal{O}\left(\frac{T^2}{\sigma^2 T_0^2}\right), \qquad
    \int_{0}^{\infty} \rd \nu \, \left|\frac{d^2 \hat q(\nu)}{d\nu^2}\right| = \mathcal{O}\left(\frac{T^{10}}{\sigma^4 T_0^6}\right).
\end{equation}

Picking $t_0 = \frac{T^4}{T_0^2 \sigma}$, this gives
\begin{equation}\begin{split}
    \Vert q(t) \Vert_{L^1} = \mathcal{O} \left(\frac{T^6}{\sigma^3 T_0^4}\right),
    \label{eq_q1normbound}
\end{split}\end{equation}

and together with Young's convolution theorem  Eq.~\eqref{eq_h-1normbound} and Eq.~\eqref{eq_h+1normbound}, $\Vert \Omega_\mathrm{DB} \Vert_\infty$ is thus bounded by~[cf. Eq.~\eqref{eq_omegaDBinfnorm}]
\begin{equation}
    \Vert \Omega_\mathrm{DB} \Vert_\infty = \mathcal{O}\left(\frac{n_B \beta}{\sigma^3} \frac{T^6}{T_0^4} e^{\frac{9\beta^2}{4\sigma^2}}\right).
\end{equation}

\subsubsection{Bounding $\Omega_\mathrm{LS}$}
The procedure to bound $\Vert \Omega_\mathrm{LS} \Vert_\infty$ [Eq.~\eqref{eq:omegaLS}] is similar to before. The same change of variables is first used to obtain a separable expression $\hat g_+(\nu_+) \, \hat g_-(\nu_-)$,
\begin{equation}
    \Omega_\mathrm{LS} = \sum_a \sum_{\nu_1,\nu_2\in\mathfrak{B}} \frac{2i}{\pi \sigma^2} \underbrace{F(\nu_-) \int_{-\infty}^\infty \rd t_- \, e^{-i\nu_- t_-} e^{-\frac{4}{\sigma^2}t_-^2}}_{\hat g_-(\nu_-)} \underbrace{\int_{-\infty}^\infty \rd t_+ \, \text{sign}(t_+) \, e^{-i\nu_+ t_+} e^{-\frac{4}{\sigma^2}\left( t_+ +\frac{i\beta}{4} \right)^2}}_{\hat g_+(\nu_+)} (A_{\nu_2}^a)^\dagger A_{\nu_1}^a.
\end{equation}

Corollary A.1 of Ref.~\onlinecite{CKG23} is again used,
\begin{equation}
    \Vert \Omega_\mathrm{LS} \Vert_\infty \leq \frac{2n_B}{\pi \sigma^2} \int_{-\infty}^{\infty} \rd t_- \, |g_-(t_-)| \int_{-\infty}^{\infty} \rd t_+ \, |g_+(t_+)|.
\end{equation}

For $\int_{-\infty}^{\infty} \rd t_+ \, |g_+(t_+)|$, it can be integrated explicitly,
\begin{equation}
    \int_{-\infty}^{\infty} \rd t_+ \, |g_+(t_+)| = \frac{\sigma}{4\sqrt\pi} e^{\frac{\beta^2}{4\sigma^2}}.
\end{equation}

For $\int_{-\infty}^{\infty} \rd t_- \, |g_-(t_-)|$, we first integrate to find $g_-(t_-)$,
\begin{equation}
    \hat g_-(\nu) = i \sqrt\pi \sigma \frac{e^{i\nu T}}{e^{-\frac{1}{4}\nu^2 T_0^2} -  e^{i\nu T}} \sinh\left( \frac{\beta\nu}{2} \right) e^{-\frac{\nu^2 \sigma^2}{16}}.
\end{equation}

Picking $\hat g_-(\nu) = i \sqrt\pi \sigma \hat r(\nu) \hat q(\nu)$,
\begin{equation}\begin{split}
    \hat r(\nu) &= \frac{1}{\nu} \sinh \left( \frac{\beta \nu}{2} \right) e^{i\nu T - \alpha \nu^2}, \\
    \hat q(\nu) &= \frac{\nu e^{-\alpha \nu^2}}{e^{-\frac{1}{4}\nu^2 T_0^2} - e^{i\nu T}}.
\end{split}\end{equation}
$\Vert g_-(t) \Vert_{L^1}$ is bounded by the convolution theorem and convolution inequality.
\begin{equation}
    \Vert g_-(t) \Vert_{L^1} \leq \sqrt{\pi} \sigma \Vert r(t) \Vert_{L^1} \Vert q(t) \Vert_{L^1}.
\end{equation}
$\hat q(\nu)$ is the same as before, so $\Vert q(t) \Vert_{L^1}$ is bounded by Eq.~\eqref{eq_q1normbound}. $\hat r(\nu)$ can be rewritten using $\frac{1}{\nu} \sinh \left( \frac{\beta \nu}{2} \right) = \int_0^{\beta/2} \rd x \, \cosh(\nu x)$, so $r(t)$ can be evaluated by swapping the integrals,
\begin{equation}
    r(t) = \frac{1}{2\pi} \int_0^{\beta/2} \rd x \int_{-\infty}^{\infty} \rd \nu \, \cosh(\nu x) e^{i\nu (T+t) - \alpha \nu^2} = \frac{1}{2\sqrt{\pi\alpha}} e^{-\frac{(T+t)^2}{4\alpha}} \int_0^{\beta/2} \rd x \, e^{\frac{x^2}{4\alpha}} \cos \left( \frac{T+t}{2\alpha}x \right).
\end{equation}
Using $\left| \cos \left( \frac{T+t}{2\alpha}x \right) \right| < 1$, $|r(t)| < \frac{\beta}{4\sqrt{\pi\alpha}} e^{-\frac{(T+t)^2}{4\alpha}} e^{\frac{\beta^2}{16\alpha}}$. The bound for $\Vert r(t) \Vert_{L^1}$ is then given by
\begin{equation}
    \Vert r(t) \Vert_{L^1} \leq \frac{\beta}{4\sqrt{\pi\alpha}} e^{\frac{\beta^2}{16\alpha}} \int_{-\infty}^{\infty} \rd t \, e^{-\frac{(T+t)^2}{4\alpha}} \leq \frac{\beta}{2} e^{\frac{\beta^2}{16\alpha}} = \frac{\beta}{2} e^{\frac{2\beta^2}{\sigma^2}}.
\end{equation}

Putting these results together, the bound for $\Vert \Omega_\mathrm{LS} \Vert_\infty$ is of the same order as $\Vert \Omega_\mathrm{DB} \Vert_\infty$.
\begin{equation}
    \Vert \Omega_\mathrm{LS} \Vert_\infty = \mathcal{O} \left( \frac{n_B \beta}{\sigma^3} \frac{T^6}{T_0^4} e^{\frac{9\beta^2}{4\sigma^2}} \right).
\end{equation}

This completes the bound for $\Vert \sigma \Vert_1$,
\begin{equation} \label{eq:sigmaboundresult}
    \Vert \sigma \Vert_1 = \Vert \Omega \Vert_\infty \leq \Vert \Omega_\mathrm{LS} \Vert_\infty + \Vert \Omega_\mathrm{DB} \Vert_\infty \leq \mathcal{O} \left( \frac{n_B \beta}{\sigma^3} \frac{T^6}{T_0^4} e^{\frac{9\beta^2}{4\sigma^2}} \right).
\end{equation}

\subsection{Bounds on $\norm{\tilde \rho-\mathcal{K}_{\mathrm{LS}}[\tilde \rho]}_1$} \label{app:boundlast}

In this subsection, we derive a bound for $\norm{\tilde \rho-\mathcal{K}_{\mathrm{LS}}[\tilde \rho]}_1$, which is the last ingredient needed to obtain the fixed-point error bound in App.~\ref{sec:Derivationbigpicture}. 

Using the expansion $\tilde \rho = \rho_\beta + J^2 \sigma$ and $e^{J^2 \mathcal L_\mathrm{LS}} = \rho + J^2 \mathcal L_\mathrm{LS}[\rho] + \mathcal{O}\left(J^4 n_B^2 e^{\frac{\beta^2}{2\sigma^2}} \right)$, this gives at order $J^4$:
\begin{align}
    \norm{\tilde \rho-\mathcal{K}_{\mathrm{LS}}[\tilde \rho]}_1 &\leq J^4 \int_{-\infty}^{\infty} \rd x \, p(x) \norm{e^{-iH(T+x)}  \mathcal{L}_{\mathrm{LS}}[\sigma]  e^{iH(T+x)} }_1+\mathcal{O}\left(J^4 n_B^2 e^{\frac{\beta^2}{2\sigma^2}} \right).
\end{align}
Due to the unitary invariance of the norm, the expression can be simplified:
\begin{align}
    \norm{\tilde \rho-\mathcal{K}_{\mathrm{LS}}[\tilde \rho]}_1 \leq J^4  \norm{\mathcal{L}_{\mathrm{LS}}[\sigma]}_1 + \mathcal{O}\left(J^4 n_B^2 e^{\frac{\beta^2}{2\sigma^2}} \right).
\end{align}

The first term is determined by collecting results for $\norm{\mathcal{L}_{\mathrm{LS}}[\rho]}_1$ [Eq.~\eqref{eq:LLSnormbound}] and $\norm{\sigma}_1$ [Eq.~\eqref{eq:sigmaboundresult}].
\begin{equation}
    \norm{\mathcal{L}_{\mathrm{LS}}[\sigma]}_1 \leq 3 n_B e^{\frac{\beta^2}{4\sigma^2}} \norm{\sigma}_1 = \mathcal{O} \left( \frac{n_B^2 \beta}{\sigma^3} \frac{T^6}{T_0^4} e^{\frac{5\beta^2}{2\sigma^2}} \right).
\end{equation}

Thus the complete bound is given by
\begin{equation}
    \norm{\tilde \rho-\mathcal{K}_{\mathrm{LS}}[\tilde \rho]}_1 = \mathcal{O} \left( J^4 n_B^2 \left( 1 + \frac{ \beta}{\sigma^3} \frac{T^6}{T_0^4} \right) e^{\frac{5\beta^2}{2\sigma^2}} \right).
\end{equation}
\section{Bounds on the variance of random channels}\label{app:boundvariance}

In this section, we provide the proof for the variance with respect to the observable $O$ after $m$ steps.

Consider a random sequence $M=(x_1,\dots,x_m)$, which corresponds to application of the random channel
\begin{align}
    \mathcal{E}_M[\rho]=\mathcal{K}_{x_m}\circ \dots \mathcal{K}_{x_1}[\rho].
\end{align}
For an arbitrary initial state $\rho$, denote the expectation value of $O$ with respect to this random sequence as $\braket{O}_M=\Tr[O \mathcal{E}_M[\rho]]$.
Furthermore, the average over all random sequences is denoted by $\mathbb{E}_M$.

The variance of the observable $O$ can then be expressed as
\begin{align}
    \mathrm{Var}(O)=\mathbb{E}_M \braket{O^2}_M-\left(\mathbb{E}_M \braket{O}_M\right)^2=
    \mathbb{E}_M\left[\braket{O^2}_M-\braket{O}_M^2\right]+\mathbb{E}_M\left[\braket{O}_M^2-\left(\mathbb{E}_M \braket{O}_M\right)^2\right].
\end{align}
In the second equation, the variance is split into two terms.
The first term is the standard fluctuation term, and is bounded by
$\left|\mathbb{E}_M\left[\braket{O^2}_M-\braket{O}_M^2\right]\right|\leq \norm{O}_\infty^2$.
The second term denotes the additional variance due to the random evolution time of the channel.

To bound the second term, note first that
\begin{align}
    \mathbb{E}_M\left[\braket{O}_M^2-\left(\mathbb{E}_M \braket{O}_M\right)^2\right]=\mathbb{E}_M\left(\braket{O}_M- \mathbb{E}_M \braket{O}_M\right)^2, 
\end{align}
and $\braket{O}_M- \mathbb{E}_M \braket{O}_M$ can be expressed as a telescopic sum:
\begin{align}
    \left[\braket{O}_M- \mathbb{E}_M \braket{O}_M\right]=\sum_{j=1}^m D_j,
\end{align}
with
\begin{align}
   D_j=\Tr[\rho \mathcal{K}^\dagger_{x_1}\circ\dots \mathcal{K}^\dagger_{x_{j-1}}(\mathcal{K}^\dagger_{x_j}-\mathcal{K}^\dagger)\circ \mathcal{K}^{\dagger(m-j)}[O]].
\end{align}
Observe that for $i>j$, 
\begin{align}
    \mathbb{E}_M[D_i D_j]=0.
\end{align}
To see this, observe that the average over the $i$-th random number $x_i$ 
contains in this case a single factor $(\mathcal{K}^\dagger_{x_i}-\mathcal{K}^\dagger)$, which averages to zero.

Furthermore, due to the contractivity of channels, it holds
\begin{align}
    |D_j|\leq \norm{(\mathcal{K}^\dagger_{x_j}-\mathcal{K}^\dagger)\circ \mathcal{K}^{\dagger{(m-j)}}[O]}_\infty.
\end{align}
Note also that $[\mathcal{K}^\dagger_{x_j}-\mathcal{K}^\dagger] \mathbb{I}=0$ due to trace preservation. We can thus add a term $-\Tr[\rho_{\rm{fix}}O]\mathbb{I}$ to the expression:
This gives 
\begin{align}
\begin{split}
    \mathbb{E}_M\left[\braket{O}_M^2-\left(\mathbb{E}_M \braket{O}_M\right)^2\right]
    &\leq \sum_{j=1}^m \mathbb{E}_M |D_j|^2\leq  \mathbb{E}_M \norm{\mathcal{K}^\dagger_{x_j}-\mathcal{K}^\dagger}^2_{\infty}\sum_{j=1}^m \norm{\mathcal{K}^{\dagger (m-j)}[O-\Tr[\rho_{\rm{fix}}O]\mathbb{I}]}_\infty^2\\
    &\leq C^2 \, \mathbb{E}_M \norm{\mathcal{K}^\dagger_{x_j}-\mathcal{K}^\dagger}^2_\infty \norm{O}^2_\infty \frac{1-e^{-\frac {2 m}{\tau_{\rm{mix}}}}}{1-e^{-\frac{2}{\tau_{\rm{mix}}}}}.
\end{split}
\end{align}
Here we used $ \norm{\mathcal{K}^{\dagger (n)}[O-\Tr[\rho_{\rm{fix}}O]\mathbb{I}]}_\infty<C e^{-\frac{n}{\tau_{\rm{mix}}}}\norm{O}_\infty$. The prefactor $C$ depends on the specific structure of the channel, such as the size of possible Jordan blocks~\cite{Szehr_2014}. 
An upper bound is given by a factor of the mixing time itself. 
However, in many cases, for example, if the channel is contractive as at high temperatures~\cite{bakshi2025dobrushinconditionquantummarkov}, $C=\mathcal{O}(1)$.

It remains to bound $\mathbb{E}_M \norm{\mathcal{K}^\dagger_{x_j}-\mathcal{K}^\dagger}^2_\infty$. We have
\begin{align}
    \mathbb{E}_M \norm{\mathcal{K}^\dagger_{x_j}-\mathcal{K}^\dagger}_\infty\leq \int \rd x\, p(x) \norm{U(x_j+T)-U(x+T)}_{\infty}\leq T_0\frac{2}{\sqrt{\pi}}\norm{H}_\infty.
\end{align}
Combining all bounds, this gives
\begin{align}
    \mathrm{Var}(O)\leq  \norm{O}^2_\infty+ C^2 T_0^2\frac{4}{\pi }\norm{H}^2_\infty \norm{O}^2_\infty \frac{1-e^{\frac {2 m}{-\tau_{\rm{mix}}}}}{1-e^{-\frac{2}{\tau_{\rm{mix}}}}},
\end{align}
which is the result in Theorem~\ref{thrm:variance}.

\section{Numerical results on other randomization schemes}\label{app:morerandom}
Here, we present numerical results on the variance of the randomization scheme introduced in Refs.~\onlinecite{ding2025endtoendefficientquantumthermal,chen2026overcominglambshiftsystembath}. These works propose a randomization of the bath Hamiltonian, and a randomization of jump operators in order to derive a tractable bound for the mixing time of the channel in terms of the exact Gibbs sampler.

For a controlled comparison with our protocol~[Fig.~\ref{fig:variance_2qubit}], we only test the randomization of the bath Hamiltonian. We use the same protocol parameters $\beta = 1, T = 10, T_0 = 1, \sigma = 1$ and the same mixed-field Ising model for $H$ [Eq.~\eqref{eq:model}] with $n_S = 2, g = 0.9045, h = 0.809$. A randomized evolution time is also used here to suppress resonances, using the same probability distribution $p(x)$ from before [Eq.~\eqref{eq:px}].

The bath Hamiltonian is determined by a random parameter $\omega$ at every interaction cycle, and is given by
\begin{equation}\label{eq:randomHB}
    H_{B,\omega} = -\frac{\omega}{2} Z,
\end{equation}
where $\omega$ is drawn from the uniform distribution $g(\omega)$ on the interval $[0, \omega_\mathrm{max}]$, with $\omega_\mathrm{max}$ set just slightly larger than the greatest spectral gap of the system Hamiltonian $H$.
\begin{equation}
    g(\omega) = \frac{1}{\omega_\mathrm{max}} \mathbbm{1}_{\omega \in [0, \omega_\mathrm{max}]}.
\end{equation}
The remainder of the Hamiltonian is similar to before,
\begin{equation}
    H_{SB, \omega}(t) = H \otimes \mathbb{I}_B + \mathbb{I}_S \otimes H_{B,\omega} + J f(t) \sum_{a\in\mathcal A} \left( A_a \otimes B_a^\dagger + A_a^\dagger \otimes B_a \right),
\end{equation}
with the filter function $f(t)$ now chosen to be real,
\begin{equation}
    f(t) = \sqrt{\frac{2}{\pi \sigma^2}} e^{-\frac{2 t^2}{\sigma^2}}.
\end{equation}
The same jump operator $A_a = X_a + Z_a$ is used. 

At the start of every interaction cycle, the bath qubits are now initialized to the bath thermal state $\rho_B^0 \propto e^{-\beta H_{B,\omega}}$. Each step of the protocol then implements a channel of the form
\begin{align}
     \mathcal{K}_{x,\omega}[\rho]=\Tr_B \left[U_\omega \left(T+x \right)(\rho \otimes \rho_B^0) U_\omega^\dagger \left(T+x \right)\right],
\end{align}
where the propagator $U_\omega\left(t+\frac{T}{2}\right)=\mathcal{T}e^{-i \int_{-\frac{T}{2}}^t \rd s \, H_{SB}(s)}$.

The channel $\mathcal{K}[\rho]$ is then given by the average of $\mathcal{K}_{x,\omega}[\rho]$ over randomized evolution time shifts $x$ and also randomized bath Hamiltonians $\omega$,
\begin{align}
    \mathcal{K}[\rho]=\int_{-\infty}^{\infty} \rd x\, \int_{0}^{\omega_\mathrm{max}} \rd \omega\, p(x) g(\omega) \mathcal{K}_{x,\omega}[\rho],
\end{align}
whereas in practice, the following random channel is implemented,
\begin{align}
    \mathcal{E}_M[\rho]=\mathcal{K}_{x_M, \omega_M}\circ \mathcal{K}_{x_{M-1}, \omega_{M-1}}\circ \cdots \circ \mathcal{K}_{x_1, \omega_1}[\rho].
\end{align}

\begin{figure*}
    \centering
    \includegraphics[width=\linewidth]{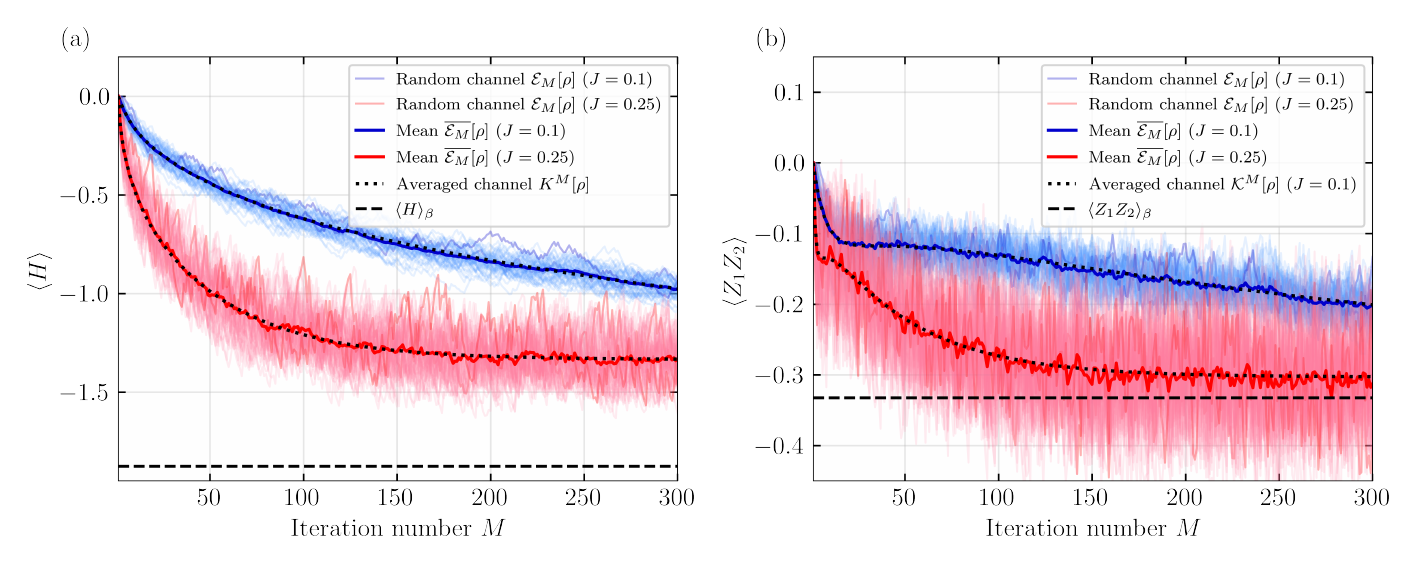}
    \caption{Evolution of thermal state preparation with a randomized bath [Eq.~\eqref{eq:randomHB}]. Results for 50 randomized channels $\mathcal{E}_M[\rho]$ (colored) and averaged channels $\mathcal{K}^M[\rho]$ (dotted) are shown for the mixed-field Ising model [Eq.~\eqref{eq:model}] with $n_S = 2,\ \beta = 1,\ \omega_\mathrm{max} = 3,\ J = 0.1$ (blue) or 0.25 (red), $T = 10$ (no resonance), initial state $\rho \propto \mathbb{I}$ (maximally mixed state). (a) Evolution of average energy $\langle H \rangle$. (b) Evolution of two-point correlator $\langle Z_1 Z_2 \rangle$. The final standard deviation is measured to be (a) $\sigma_{H} = 0.040 \ (J=0.1),\ 0.106\ (J=0.25)$, (b) $\sigma_{Z_1 Z_2} = 0.024\ (J=0.1),\ 0.060\ (J=0.25)$, which is greater than that of our protocol [Fig. \ref{fig:variance_2qubit}].}
    \label{fig:linlin6}
\end{figure*}

Since this channel now involves two random sequences of variables $x$ and $\omega$, it is naturally expected that this protocol has a larger variance. The numerical results are plotted in Fig.~\ref{fig:linlin6}, which indeed show that the fluctuations are more significant than our protocol [Fig. \ref{fig:variance_2qubit}], with $\sigma_{H}$ larger by a factor of 3--5 and $\sigma_{Z_1 Z_2}$ larger by a factor of 1.18--1.26.

\end{widetext}
\end{appendix}

\twocolumngrid

\bibliography{references_paper}

\end{document}

%% file: preamble.tex
\usepackage[utf8]{inputenc}
\usepackage{fullpage}
\usepackage{graphicx}
\usepackage{amsmath}
\usepackage{amsthm}
\usepackage{enumitem}
\usepackage{amssymb}
\usepackage{wasysym}
\usepackage{oplotsymbl}
\usepackage{bm} 
\usepackage{dsfont} 
\usepackage{mathrsfs} 
\usepackage{physics}
\usepackage{mathtools}
\usepackage{quantikz}
\usepackage{braket}
\usepackage[dvipsnames]
{xcolor}
\usepackage{bbm} 
\definecolor{darkblue}{rgb}{0,0,.65}
\definecolor{darkgreen}{rgb}{0.28,0.41,0.19}
\definecolor{nicegreen}{rgb}{0.28,0.85,0.19}

\usepackage[colorlinks=true,linkcolor=blue,allcolors=blue]{hyperref}
\usepackage[capitalize,nameinlink]{cleveref}

\usepackage{comment}
\usepackage[normalem]{ulem}

\def\equationautorefname~#1\null{Eq. (#1)\null}

\newcommand{\appref}[1]{\hyperref[#1]{App.~\ref*{#1}}}

\newcommand{\e}{\mathrm{e}}

\newcommand{\rd}{\mathrm{d}}

\newtheorem{theorem}{Theorem}
\newtheorem{lemma}{Lemma}